\documentclass[12pt]{article}  

\usepackage{titlesec}       
\titlespacing{\section}{0pt}{-0.2ex plus .1ex minus .0ex}{-0.1ex plus .1ex} 
\titlespacing{\subsection}{0pt}{-.2ex plus .1ex}{-.2ex plus .1ex}  

\usepackage{algorithm}   
\usepackage{algorithmic}   

\usepackage{amsmath,amssymb}    
\makeatletter    
\g@addto@macro\normalsize{   
	\setlength\abovedisplayskip{2mm}          
	\setlength\belowdisplayskip{2mm}         
	\setlength\abovedisplayshortskip{0mm}    
	\setlength\belowdisplayshortskip{0mm}     
}        
\makeatother

\usepackage{amstext}   
\usepackage[strings]{underscore}  
\usepackage{bm}  
\usepackage{float}   
\usepackage{booktabs}

\makeatletter
\let\c@lofdepth\relax 
\let\c@lotdepth\relax
\makeatother
\usepackage{subfigure}  

\usepackage{mathtools}   

\usepackage{CJK}   
\usepackage{makecell}     
\usepackage{fancyvrb}      
\usepackage{upgreek}     
\usepackage{diagbox}    
\usepackage{color}      
\usepackage{natbib}       
\usepackage{graphicx}     
\usepackage{hyperref}  
\usepackage{setspace}      
\usepackage{ntheorem} 
\usepackage{paralist} 
\usepackage{appendix}
\usepackage{hyperref}       
\usepackage{threeparttable} 

\newtheorem{mydef}{Definition}  
\newtheorem{mythm}{Theorem} 
\newtheorem{mylemma}{Lemma} 
\newtheorem{myprop}{Proposition}   
\newtheorem*{proof}{Proof}

\begin{document}

\noindent  \Large On a Novel Skewed Generalized $t$ Distribution:
Properties,  Estimations and its Applications.   
 
\vskip 3mm     
\normalsize     

\vskip 3mm

\vskip 5mm  
\noindent  Chengdi Lian, Yaohua Rong and Weihu Cheng*   

\noindent Faculty of Science  

\noindent Beijing University of Technology

\noindent chengweihu@bjut.edu.cn   

\vskip 3mm
\noindent Key Words:   skewed  generalized $t$ distribution,  maximum likelihood estimation,  L-moments estimation, two-step estimation,   EM algorithm  
\vskip 3mm  

\doublespacing

\noindent ABSTRACT
 
With the progress of information technology, large amounts of asymmetric, leptokurtic and heavy-tailed data are arising in various fields, such as finance, engineering, genetics and medicine.         
It is very challenging to model those kinds of data,  especially for extremely skewed data, accompanied by very high kurtosis or heavy tails.             
In this paper,  we propose a class of novel skewed generalized $t$ distribution (SkeGTD) as a scale mixture of skewed generalized normal.    
The proposed SkeGTD has excellent adaptiveness to various data, because of its capability of allowing for a large range of skewness and kurtosis and its compatibility of the separated location, scale, skewness and shape parameters.             
We investigate some important properties of this family of distributions.   
The maximum likelihood estimation, L-moments estimation and two-step estimation for the SkeGTD are explored.  To illustrate the usefulness of the proposed methodology, we present simulation studies and analyze two real datasets.

\section{Introduction \label{introduction} }    
   
Skewed,  leptokurtic and heavy-tailed data have become more and more widespread with the progress of information technology, especially in several areas like finance, engineering and medicine.                  
Thus,  the traditional normality hypothesis may be inappropriate.       
The families of skewed distributions  will be logical candidates, such as skew-normal, skewed logistic and skewed $t$ distributions 
(\citet{9465733,balakrishnan2012multivariate,asgharzadeh2013approximate,arslan2009skew,garcia2010new}).         
The skewed normal distribution was first proposed by      \citet{azzalini1985class}.       
However, when dealing with extremely skewed data, especially accompanied by very high kurtosis and heavy tails,  the behavior of the skew-normal (SN) distribution  may be unsatisfactory.           

To this end,   \citet{theodossiou1998financial}  proposed the skewed generalized $t$ distribution (SGTD) to fit the  empirical distribution of financial series,  then \citet{venegas2012robust} gave its stochastic representation, Fisher information matrix, and some crucial properties.    
Following the pioneering  work of \citet{theodossiou1998financial}, many  different types of  SGTD  have been proposed.      
\citet{zhu2012asymmetric}  proposed the asymmetric generalized $t$ distribution and presented some properties of the maximum likelihood estimation (MLE).    
The multivariate extensions of the SGTD  were studied  by \citet{azzalini2003distributions},  \citet{azzalini2008robust}, \citet{ lin2010robust} and \citet{lee2014finite}, respectively.      
The SGTD has been extensively used over the last decades for modeling data in different research areas such as  regression models   
(\citet{hansen2006partially}),  financial econometrics (\citet{zhu2010generalized}), and environmental sciences (\citet{acitas2015alpha}).    

Those skewed versions of the SGTD were proposed by using different skewing mechanisms. The most common approach for generating skewed distributions is to introduce a skewness parameter, denoted as $\lambda$ ($\lambda \in [-1, 1]$), into the probability density function (pdf) of a symmetric distribution. This parameter is used to regulate both the skewness and the heaviness of the tails (\citet{ferreira2006constructive}).   
The skewed extension constructed by this approach has only one parameter to control the behavior of left and right tails, but it does not include another parameter that allows skewness to change independently of the tail parameter.  
Moreover,  the adjustable range of $\lambda$  is often limited to a finite interval, which means the existing SGTDs may not suffice to adequately handle different kinds of real data,  especially the skewed data with very high kurtosis or heavy tails.

In order to overcome this limitation,  we coin a new class of  SkeGTD models using the scale mixture approach,  which  allows  independent regulation of  location, scale, skewness and heaviness of the tails.                
The SkeGTD family nests many common  distributions as its special cases, such as Student-t,  Pareto,  Normal, Cauchy, Laplace distributions and so on.   
The main feature of the SkeGTD is that it allows for a larger range of skewness and kurtosis than existing  skewed $t$ distributions and can fit the data as adequately as possible. Thus, the proposed SkeGTD exhibits excellent adaptability to various types of data and is expected to find wide application across diverse fields. 
 
The first target of this paper is to construct the SkeGTD by applying a suitable random scale factor to the skewed generalized normal (SGN) distribution.          
In this way,   we can conveniently obtain the stochastic  representation,  which is useful to generate random numbers of SkeGTD.       	        
The second target is to discuss inferential and other statistical issues.      
Different parameter estimation methods including the MLE and the L-moments estimation (LME) are proposed for the three-parameter SkeGTD.      	 
In particular, an EM-type algorithm can be implemented based on the construction of SkeGTD, and the expressions of the Fisher information matrix are given.               
We also propose a two-step estimation (TSE) for the five-parameter SkeGTD.

The rest of this paper is organized as follows.          
In Section  \ref{sec2},  we propose the SkeGTD and investigate its several appealing properties.  
In Section  \ref{sec3}, some theoretical properties and estimation methods  are explored.   
In Section \ref{sec4},  we examine the finite-sample  performance of the proposed estimation methods by simulation.  
In Section \ref{sec5},  the analyses of two real data sets are performed to illustrate the flexibility and applicability  of the SkeGTD.              
In Section \ref{sec6},  the paper ends  with a brief conclusion,  and  detailed proofs are given in the Appendix.

\section{The Novel Skewed Generalized $t$ Distribution    }   \label{sec2}

\citet{azzalini2003distributions} proposed the skewed $t$ distribution by mixing the SN random variable (\citet{azzalini1996multivariate}) with the inverse gamma random variable.   
Let $ Y \sim SN(\lambda)$ and $ Z \sim Ga (\nu/2, \nu/2) $ be independent random variables,  where $Ga(\alpha, \beta)$ denotes   the  gamma  distribution   with mean value $ \alpha / \beta$  and variance $ \alpha / \beta^2$. Then  the location-scale version of the Azzalini-type skewed $t$ distribution can be written as $ X = \mu + \sigma Y / \sqrt{Z}   , $     
denoted by $ X \sim St(\mu, \sigma, \lambda, \nu)$. We proposed here the new SkeGTD inspired by the work of \citet{azzalini2003distributions}, and its pdf is given as follows.   
 
\subsection{Definition and  Stochastic  Representation }

\begin{mydef} \label{defsgt}  
	Let $X$ be a continuous random variable with the density     
	\begin{equation} \label{fsgt} 
		f_{SkeGTD} (x; \mu, \sigma , r, \alpha, \beta)  =  \frac{\beta} { 2\sigma(2\alpha)^{1/{\beta}}{B(\alpha,1/{\beta})} }   
		\left\{   1+\frac{|x-\mu|^{\beta}}{2\alpha \sigma ^\beta [ 1+r {\rm sign} (x- \mu )]^ { \beta}  } \right\} ^{-(\alpha+1/\beta)}, 
	\end{equation}	 
	\noindent where $ x \in \mathbb{R}$, $ B(\cdot, \cdot)$ denotes the beta function,   
	$\mu \in \mathbb{R} $, $\sigma >0 $ 	and $ |r| \le 1 $ determine location, scale and  skewness, respectively,  while  $\alpha >0 $ and $\beta>0 $ control the shape of the density function.     
	Then the random variable $X$ follows a SkeGTD, denoted by $ X \sim SkeGTD (\mu, \sigma, r, \alpha, \beta). $   
	The distribution is positively (negatively) skewed for $ r > 0~  (< 0) $.             
	The symmetric generalized $t$ distribution  proposed by \citet{math9192413} can be obtained for $ r=0$.           	            
\end{mydef}       

Next, we show that the SkeGTD  is a scale mixture of the following SGN distribution and a gamma distribution.          
A  random variable $X$ follows the SGN distribution with location parameter  $ \mu \in \mathbb{R} $,  scale parameter  $\sigma>0 $, skewness parameter $ |r| \le 1 $  and shape parameter $\beta>0 $, denoted as $X \sim SGN(\mu, \sigma , r, \beta) $,  if it has the density   
\begin{equation}  
	f_{SGN} (x; \mu, \sigma , r, \beta)=\frac{\beta}{2^{1+1/{\beta}} \Gamma(1/{\beta}) \sigma } \exp \left\{ -\frac{|x-\mu|^{\beta}}{2 \sigma ^\beta [ 1+ r {\rm sign} (x-\mu) ]^\beta } \right\}, \ \   x \in \mathbb{R}. \label{ fsgn }  
\end{equation}    

\begin{mythm} \label{thmsgt.repre}  
	
	Let  $ K \sim SGN(0, 1 , r, \beta) $, $Z \sim  Ga(\alpha,\alpha)$,  $K $ and $ Z$ are independent,  and  
	\begin{equation}  
		X= \mu +  \sigma  K Z^{ -1/ \beta} , \label{sgt.repre} 
	\end{equation} 
	then  $ X \sim SkeGTD (\mu, \sigma, r, \alpha, \beta)$ with the pdf given by (\ref{fsgt}), 
	\noindent where the parameters   $ \mu \in \mathbb{R}, \sigma > 0, |r| \le 1,\alpha>0 $ and $  \beta>0.$    
\end{mythm}

\begin{proof} 	 
	A random variable $ X_0 = (X-\mu) / \sigma $   is distributed as the normalized SkeGTD.   
	It suffices to verify the pdf of $ X_0 $.     
	Let $ U =  Z^{ 1/ \beta} $ and 
	use the transformation $ v=  u ^ \beta \varrho $ for the following integral   
	\begin{eqnarray}  
		f _{ X_0 } (x) &=& \  \int_{0}^{\infty}   
		f _K (xu) f _U (u) u\  \mathrm{d} u 
		=  \  \frac{\beta^2 \alpha ^ \alpha }
		{2 ^ {1 + 1 / \beta} \Gamma (\alpha)  \Gamma (1 / \beta) }
		\int_{0}^{\infty} u ^ {\alpha \beta } 
		\exp \{  - u ^ \beta \varrho   \}  \  \mathrm{d} u,    \notag  \\
		&=& \  \frac{\beta}
		{   2 ^{1 + 1/{\beta} }  \alpha ^{1/{\beta}} B(\alpha,1/{\beta}) }       
		\left\{ 1+\frac{|x|^{\beta}}{2\alpha  [ 1+r {\rm sign}(x) ]^ { \beta}  } \right\}^  {-(\alpha+1/\beta)} , \  \ x \in \mathbb{R}, 	 \label{fx0}  
	\end{eqnarray}   
	where $ \varrho= \alpha + |x| ^ \beta /  \{  2 [ 1 + r {\rm sign}(x) ]   ^ \beta \}   $. Then we can obtain the pdf of  $  X_0  $.    
\end{proof}

To illustrate the effect of $  r, \alpha,  $ and $\beta$ on the distribution,   we plot different density curves of the normalized SkeGTD in Figure~\ref{plotf}.            
Figure~\ref{plotf}(a)  shows  a symmetric distribution $(r = 0)$ and its skewed counterparts for $r = -0.9, -0.5, $  and 0.5, respectively,  with $\alpha = \beta = 2$.    
From Figure~\ref{plotf}(a), it is observed that 
the SkeGTD has a right skewness for $ r > 0 $ and a left skewness for  $ r < 0 $, respectively.       
Note that when the sign of $r$ is reversed,  the density is mirrored on the opposite side of the vertical axis  and the skewness increases with the value of $|r|$.     
The peakedness and tail heaviness of the density curve is controlled by $\alpha $ and $\beta$ simultaneously.    
The two parameters are included for different purposes.      
Figure~\ref{plotf}(b) shows that,   given $r$ and $\beta$, the density becomes leptokurtic  as $\alpha $  increases.  
We  notice that  the parameter $\alpha$ mainly regulates the tail behavior of the density.     
Figure~\ref{plotf}(c) shows how the shape of the density curve changes from steep to smooth  when  $\beta$  is from small to large.     
The parameter $\beta$ plays a decisive role in the shape when $x$ is near the location and   
larger values of $\beta $ are associated with flatter peaks and thinner tails of the density.

\subsection{Properties}

Next, we give some theoretical properties of the SkeGTD, including the stochastic representation,  moments, skewness and kurtosis coefficients, and cumulative distribution function (cdf).     
We begin by discussing the stochastic representation of the SGN distribution.

\begin{mylemma}\label{sgnrandom}   
	If $ X \sim  SGN(\mu, \sigma , r, \beta) $, then $ X \stackrel{d}{=} \mu +  \sigma  2^{1/ \beta}  W Y^{1/ \beta} $, where $Y \sim Ga(1/\beta, 1) $ and $W$ is an independent discrete variable which satisfies
	\begin{equation}
		P \left(  W=r+1 \right)  =  (r+1)/2 , ~~ P \left( W=r-1 \right)  = (1-r)/2 ,   
	\end{equation}   
	where $  -\infty< \mu <\infty, \sigma>0, |r| \le 1,\alpha>0 $ and $\beta>0. $ 
\end{mylemma}        

Note that if  $x \ge \mu$,   	    
	\begin{equation}                  
	 P\left( X\le x \right) = \  \frac{r+1}{2}  P\left\{ (r+1) Y^{1/ \beta}   \le (x- \mu)/\sigma  \right\} + \frac{1-r}{2}.    \label{pxx}  
	\end{equation}   
	
	\noindent We have the density (\ref{ fsgn }) by the derivative of Equation (\ref{pxx})  with respect to $x$.       
	The case of   $ x \le \mu $ will be obtained as claimed.     

Proposition \ref{SGTrandom} presents additional results for SkeGTD, which are valuable for sample generation and the development of the EM-type algorithm.

\begin{myprop} \label{SGTrandom}  
	Let $ X \sim SkeGTD(\mu, \sigma , r, \alpha, \beta) $, then 	
	\begin{equation}  
		X \stackrel{d}{=} \mu +  \sigma 2^{1/ \beta}  W Y^{1/ \beta}  Z^{ 1/ \beta}, \label{structure}  
	\end{equation}   
	\noindent where $W, Y$ and $Z$ are independent random variables, $ W $ and $ Y $ are defined as in Lemma \ref{sgnrandom}, and $ Z \sim IG(\alpha,\alpha) $, where $ IG(\cdot) $ denotes the inverse gamma distribution.      
\end{myprop}   
Proposition \ref{SGTrandom} can be intuitively derived from Theorem \ref{thmsgt.repre} and Lemma \ref{sgnrandom}. This property offers advantages for inferring moments of $X$ and generating random numbers. Now let us consider the moments of SkeGTD. Let $ X_0 = (X-\mu) / \sigma $, for any positive integer $  k < \alpha\beta  $, then 
\begin{equation} 
	E(X^k) 
	= \sum_{i=0}^{k} \binom {k} {i}  \sigma^{i} \mu^{k-i} E( X_0 ^i) , \label{Esgt} 
\end{equation}  
where 	$  E( X_0 ^i)  $ is given by    
\begin{equation}
	E( X_0 ^i) 
	= (2 \alpha)^ {i/\beta}  	\frac{\Gamma (\alpha- i/\beta ) \Gamma ( (i+1)/\beta )}{2 \Gamma(\alpha)  \Gamma(1/\beta)} \left[  (r+1)^{i+1} - (r-1)^{i+1} \right].  \label{Exok}  
\end{equation}  
Formula (\ref{Exok}) is derived using 
$ E(X_0^i)  =    2^{i/ \beta} E(W^i)  E( Y^{i/ \beta} ) E( Z^{ i/ \beta} ) $. 	       
With some basic algebraic manipulations,  we have        
	\begin{align}   
		E(X; \mu, \sigma, r, \alpha, \beta) =& \ \mu + 2 (2 \alpha)^ {1/\beta} \sigma  r	\frac{B (\alpha- 1/\beta, 2/\beta)}{ B(\alpha,1/\beta)} ,   	\\                 
		Var(X; \sigma, r, \alpha, \beta) =& \        
		\frac{ \sigma ^2 (2 \alpha)^ {2/\beta}}{  B(\alpha,1/\beta) } \left\{  (3 r ^2 +1 ) B (\alpha- 2/\beta, 3/\beta) - \frac{4 r^2 B^2 (\alpha- 1/\beta, 2/\beta)} {B (\alpha, 1/\beta)}\right\} ,  \label{sgtvar}    
          \\   
		\gamma_1(X; r,\alpha,\beta) 
		=&  \ \frac{ 2 r  B ^ {1/2} (\alpha, 1/\beta) }{ [H(r, \alpha, \beta)] ^{3/2}  }   \left\{    r^2 H_1(\alpha, \beta)+ H_2(\alpha, \beta)  \right\} ,    \label{g1}  
		\\    
		\gamma_2(X; r,\alpha,\beta) 
		=&  \  \frac{ B(\alpha,  1/\beta )  }{ [H(r, \alpha, \beta)] ^{2}  }   \left\{     r^4 H_3(\alpha, \beta)+  r^2 H_4(\alpha, \beta) + B(\alpha- 4/\beta ,  5/\beta)    \right\}       
		-3 ,      \label{g2}           
	\end{align}  
in the cases of $ \alpha \beta >1, ~ \alpha \beta >2,~ \alpha \beta >3$ and $  \alpha \beta >4$,  where    
	\begin{align}      
		H(r, \alpha, \beta) = \  (& 3 r^2 +1 ) B(\alpha-2/\beta, 3/\beta) -\frac{ 4 r^2  B ^ {2} (\alpha-1/\beta , 2/\beta)   }{B(\alpha, 1/\beta )}  , 
		\notag \\  
		H_1(\alpha, \beta) =  \   2&  B(\alpha-3/\beta , 4/\beta) - 9 \frac{B(\alpha-1/\beta , 2/\beta)\ B(\alpha-2/\beta , 3/\beta)}{B(\alpha , 1/\beta)}   
		+  8 \frac{ B  ^ 3 (\alpha-1/\beta , 2/\beta)  }{ B ^2 (\alpha-1/\beta , 2/\beta)  } ,\notag \\
		H_2(\alpha, \beta) =  \  2&  B(\alpha-3/\beta , 4/\beta) - 3 \frac{B(\alpha-1/\beta , 2/\beta)\ B(\alpha-2/\beta , 3/\beta)} {B(\alpha , 1/\beta)},  \notag  \\
		H_3(\alpha, \beta) =  \  5&  B(\alpha-4/\beta , 5/\beta) - 2^5 \frac{B(\alpha-1/\beta , 2/\beta)\ B(\alpha-3/\beta , 4/\beta)}{B(\alpha , 1/\beta)}   \notag  \\ 
		+&  3^2 \cdot 2^3\  \frac{  B  ^ 2 (\alpha-1/\beta , 2/\beta) B(\alpha-2/\beta , 3/\beta) }{  B  ^ 2 (\alpha, 1/\beta)  }  -3\cdot2^4 \frac{ B ^ 4 (\alpha-1/\beta , 2/\beta)  }{ B ^ 3 (\alpha, 1/\beta) }, \notag \\   
		H_4(\alpha, \beta) = \   10&  B(\alpha-4/\beta , 5/\beta) - 2^5 \frac{B(\alpha-1/\beta , 2/\beta)\ B(\alpha-3/\beta , 4/\beta)}{B(\alpha , 1/\beta)}    \notag  \\ 
		+&  3 \cdot 2^3\  \frac{  B ^ 2 (\alpha-1/\beta , 2/\beta)  B(\alpha-2/\beta , 3/\beta) }{ B ^ 2 (\alpha, 1/\beta) }, \notag 		 
	\end{align} 
 and $\gamma_1(\cdot)$ and $\gamma_2(\cdot)$ are the measures of skewness and kurtosis, respectively.

Note that   
$	\gamma_i(X;r,\alpha,\beta) =  \  \gamma_i(X_0;r,\alpha,\beta) $,
$ i=1, 2  $.  Formula (\ref{Exok}) allows us to derive formulas (\ref{g1}) and (\ref{g2}). 
The following properties of $\gamma_1(\cdot)$ and $\gamma_2(\cdot)$ can be numerically verified:         
(a) Both $\gamma_1(\cdot)$ and $\gamma_2(\cdot)$ monotonically decrease with respect to $\alpha$ and $\beta$ for any given $r > 0$ and increase monotonically with respect to $|r|$ for given $\alpha$ and $\beta > 0 $.        
(b)  $\gamma_1(X;r,\alpha,\beta) = -\gamma_1(X;-r,\alpha,\beta)  $, 
$\gamma_2(X;r,\alpha,\beta) =  \gamma_2(X;-r,\alpha,\beta)  $,   
$  \lim\limits_{ \alpha \beta  \rightarrow 3  } |\gamma_1(X;r,\alpha,\beta) | =  + \infty $ ,   
$ \lim\limits_{  \alpha \beta  \rightarrow 4 }  \gamma_2(X;r,\alpha,\beta)  = + \infty $,  
and 
$ \lim\limits_{  \beta  \rightarrow +\infty }  \gamma_2(X;r,\alpha,\beta)  = -1.2  $.                
Both  $ \gamma_1(\cdot)$ and $ \gamma_2(\cdot)$ have wide ranges.  
Therefore, the SkeGTD can capture a broader range of asymmetric and leptokurtic behaviors,  making it a more flexible option for real data analysis.

By integrating  density (\ref{fx0})  with respect to $x$ and using the transformation $  u  =  \{ 1+  |x| ^ \beta / \{  2 \alpha  \left[   1+ r {\rm sign}(x) \right] ^ \beta \}  \}  ^ {-1} $,   the cdf of SkeGTD can be given by       
	\begin{eqnarray}  \label{FX}     
		\begin{aligned}
			F_{X_0} (x; r, \alpha, \beta ) = 
			\begin{dcases}  
				\left( \frac{1-r}{2} \right) I_{u(x; r, \alpha, \beta)} \left( \alpha, \frac{1} { \beta }  \right) , & x\le  0,  \\
				1-\left( \frac{1+r}{2} \right) I_{u(x; r, \alpha, \beta)}  \left( \alpha, \frac{1} { \beta }  \right) , & x > 0, 
			\end{dcases}	 	 
		\end{aligned}
	\end{eqnarray}           
	\noindent where $ I_y (a,b) = \frac{1}{ B(a,b)} \int_{0}^{y}\ t^{a-1} (1-t)^{b-1} \mathrm{d}t  $ denotes the  incomplete beta function and    
	\( 
	u(x; r, \alpha, \beta) =  \{    1+  |x| ^ \beta / \{  2 \alpha  \left[   1+ r {\rm sign}(x) \right] ^ \beta \}  \}  ^ {-1}.      
	\)

Further, we can derive an explicit expression for the cdf of a random variable $X$ by using   $  F_{X} (x; \mu, \sigma, r, \alpha, \beta ) = F_{ X_0} ((x- \mu )/\sigma;  r, \alpha, \beta )    $.        
The special and limiting cases of the $ SkeGTD(\mu, \sigma, r, \alpha, \beta)  $   are given as follows.

\begin{myprop} \label{ppp4}
	
	Let $ X \sim SkeGTD(\mu, \sigma, r, \alpha, \beta) $, then  the following results hold:   

	\begin{compactitem}      
		\item[(1).]   
		If $  \alpha \rightarrow \infty $,  then $ X \rightarrow SGN(\mu, \sigma, r,  \beta)$;
		\item[(2).]   
		If $ \mu=0,~ \sigma=1,~ r=0, ~ \alpha= n/2 $ and $ \beta=2 $,  then   $ X \sim t(n)$; 
		\item[(3).]   
		If $ r=1\  and\   \beta=1 $, then $ X \sim  Pareto(\uppercase\expandafter{\romannumeral2}) ( \mu, 2\sigma\alpha, \alpha )$;  
		\item[(4).]   
		If $ \alpha \rightarrow \infty,~ r=0 $ and $ \beta =2 $ ,  then 
		$ X \sim N(\mu, \sigma^2)$; 
		\item[(5).] 
		If  $ \alpha \rightarrow \infty,~ r=0 $  and $ \beta=1 $ , then 
		$ X \sim  Laplace (\mu, 2 \sigma )$;   
		\item[(6).]  
		If  $  \alpha \rightarrow \infty,~  r=0 $  and $ \beta \rightarrow \infty $ ,  then  
		$ X \sim U  [\ \mu-\sigma , \mu+\sigma \ ]$;   
		\item[(7).]      
		If  $  r=0, ~\alpha= 1/2  $  and $ \beta =2  $,  then 
		$ X \sim Cauchy(\mu, \sigma) $.                 
	\end{compactitem}     
\end{myprop} 

\begin{proof}   
	The proof is given in Appendix \ref{appen.prop2}. 
\end{proof}

\section{Parameter Estimation }  \label{sec3}   

In this section, for the normalized SkeGTD ($\mu=0$, $\sigma=1$), we first obtain the parameter estimators using the likelihood method and derive the fisher information matrix. To enhance the accuracy of shape parameters estimation with small sample sizes, we also investigate L-moments estimation. Additionally, we introduce a two-step estimation method for the five-parameter SkeGTD.

\subsection{ Maximum Likelihood Estimation }       
            
Now, let's consider the MLE of the parameters of the normalized SkeGTD using the Expectation-Maximization (EM) algorithm. The EM algorithm is commonly employed to determine MLE, and it is a powerful computational technique, especially useful in situations involving missing data. In summary, the algorithm consists of two steps: the E-step, which involves computing the conditional expectation of the complete-data log-likelihood, and the M-step, which maximizes the expected value with respect to the unknown parameters. This process iterates between the two steps until a convergence criterion is met (\citet{chen1999new, aas2006generalized, lin2007robust}).  

Let $y_1, \ldots, y_n$ be a random sample of size $n$ from $SkeGTD(0, 1, r, \alpha, \beta)$, and let $\boldsymbol{\theta} = (r, \alpha, \beta)^{T}$. Based on the stochastic representation in Theorem \ref{thmsgt.repre},  we have $y_i \stackrel{d}{=} x_i z_i^{-1/\beta}$, where $x_i \sim SGN(0,1, r, \beta)$ are independent of $z_i \sim Ga(\alpha, \alpha)$,  for $i=1, \ldots, n$. We consider $\boldsymbol{y} = (y_1, \ldots, y_n)^{T}$ as the observed data, while $\boldsymbol{z} = (z_1, \ldots, z_n)^{T}$ as the missing data. Note the following hierarchical representation:            
\begin{equation}
	y_i~|~ z_i \sim SGN(0, z_i ^{-1/ \beta } , r, \beta) ,   ~~~
	z_i    \sim  Ga(\alpha, \alpha ) ,  \ \ \ \ \  \ \ 
	i = 1,  \dots , n.     \label{represent}    
\end{equation}   
From (\ref{represent}), the joint pdf of $ (y_i, z_i ) $ is given by
	\begin{equation}  
	f_{ y,z } \left( y, z|\ \boldsymbol{\theta}  \right)
	\propto  \ 
	z^{ \alpha + 1/ \beta -1 }  \exp \left\{  -z \left[    \alpha +  \frac{ |y|^{\beta} }  {2 \alpha (1+r {\rm sign}(y))^ { \beta} }  \right] \right\},  \label{fyz}    
\end{equation} 
\noindent where we omit the index $i$ for brevity.     
Dividing (\ref{fyz}) by (\ref{fx0}), we obtain the conditional distribution of $z_i$ given $\boldsymbol{\theta}^t$ and $y_i$ as 
\begin{equation}  
	z_i \ |(\boldsymbol{\theta}^t , y_i) \sim Ga( a^t, b_i^t ) \label{zcy} ,       
\end{equation}   
where 
$ 
a^t =\ \alpha^t + 1/ \beta^t ,   ~~~   
b_i ^t =\  \alpha^t +  |y_i|^{\beta^t} \big/ \{ 2  \alpha^t [ 1+r^t {\rm sign}(y_i)] ^ { \beta^t}  \}  ,     
$
and  $ \boldsymbol{\theta} ^{t} $ denotes the current estimate  $\boldsymbol{\theta}$ at the $t$-th step.

Then the complete likelihood for $ (\boldsymbol{y}, \boldsymbol{z})$ is 
\begin{equation} 
	\begin{split}
		\ell( \boldsymbol{\theta}; \boldsymbol y, \boldsymbol z)  
		= 
		n& \log \left\{   \frac{ \alpha ^ \alpha \beta}
		{  2^ {1 + 1/ \beta}  
			\Gamma(\alpha)  \Gamma (1/ \beta) }  \right\}   
		+  (\alpha + 1/\beta -1 )  \sum_{i=1}^{n} \log \left( z_i \right)  
		  \\ 
		-&   \  \sum_{i=1}^{n} z_i  \left\{    \alpha + \frac{ |y_i|^{\beta} }  {2 \left[   1+ r {\rm sign}(y_i) \right] ^ { \beta} }  \right\}  , 
		\ \ \ ~~~  y_i \in \mathbb{R},~~  z_i >0, \ i=1 ,\dots n.   \label{logfyz}  
	\end{split}  
\end{equation}

 It is worth noting that the above maximizer does not have explicit solutions for $\boldsymbol{\theta}$. To address this, we develop an EM algorithm to simplify the computation.  
 
 In the E-step, given the current estimate $\boldsymbol{\theta}^t$, we calculate $E(\ell(\boldsymbol{\theta}; \boldsymbol{y}, \boldsymbol{z}))$, denoted as $Q(\boldsymbol{\theta} | \boldsymbol{\theta}^t, \boldsymbol{y})$. This simplifies to the computation of $E(z_i | \boldsymbol{\theta}^t, y_i)$ and $E(\log(z_i) | \boldsymbol{\theta}^t, y_i) $. Note that    
\begin{equation}
	E\left( z_i \ | \boldsymbol{\theta}^t , y_i \right)   = \ a^t / b_i ^t,   ~~ 	E\left( \log (z_i) \  | \boldsymbol{\theta}^t ,  y_i \right)  = \   \psi (a^t) - \log ( b_i ^t ),    \notag  
\end{equation}	
\noindent  where  $ \psi(\cdot) $ is the digamma function.  After some algebraic manipulations, we have    
\begin{equation}  
	\begin{split}
	Q\left( \boldsymbol{\theta}|\ \boldsymbol{\theta}^t ,\boldsymbol{y} \right) = \ 
	n& \log \left(  \frac{ \alpha ^ \alpha \beta }  { 2^{1+1/{\beta}} \Gamma (\alpha) \Gamma ( 1/\beta )  } \right) + (\alpha+1/\beta-1 ) 
	\sum_{i=1}^{n} \left\{  \psi (a^t) - \log (b_i^t )  \right\}  
	\notag  \\     
	-& \sum_{i=1}^{n}  \left\{   \frac{	a^t} { b_i^t}  
	\left[   \alpha  + \frac{ |y_i|^{\beta} } { 2 (1+r {\rm sign}(y_i) )^ \beta}   \right]  
	\right\} .   	\label{Q}   
	\end{split}     
\end{equation}

Next, in the M-step, we find the maximizer of $Q\left(\boldsymbol{\omega} | \boldsymbol{\omega}^t, \boldsymbol{y}\right)$ to update the estimates. To avoid a breakdown of the EM algorithm, let $\eta = -1/\beta$ and $\boldsymbol{\omega} = (r, \alpha, \eta)$, and substitute $Q\left(\boldsymbol{\omega} | \boldsymbol{\omega}^t, \boldsymbol{y}\right)$ for $Q\left(\boldsymbol{\theta} | \boldsymbol{\theta}^t, \boldsymbol{y}\right) $. By taking partial derivatives of $Q\left(\boldsymbol{\omega} | \boldsymbol{\omega}^t, \boldsymbol{y}\right)$ with respect to $r, \alpha, \eta$, and setting them to zero, we obtain       
\begin{align}
	\left\{  
	\begin{aligned}
		\ \ &  \sum_{i=1}^{n} \left\{  \frac{|y_i| ^ {1/ \eta }  {\rm sign}(y_i) }  { b_i^t [  1+r  {\rm sign}(y_i) ]^ { 1/ \eta  +1 } }    \right\}  =\ 0 ,  
		\\  \label{t+2} 
		\ \ &  \log (\alpha) +1 - \psi (\alpha) + \psi (a^t) - \frac{1}{n} \sum_{i=1}^{n}  \log (b_i^t) - \frac{a^t}{n}  \sum_{i=1}^{n}  \frac{1} { b_i^t } =\ 0  ,   
		\\  
		\ \ &  1/ \eta + \log 2  + \psi (\eta) - \psi ( a^t ) + \frac{1}{n}  \sum_{i=1}^{n}  \log (b_i^t)  
		- \frac{a^t} { 2n \eta ^2 } \sum_{i=1}^{n}  \left\{  \frac{ M_i ^ { 1/ \eta } }  { b_i^t   }  \log \left(  M_i   \right)  \right\}   =\ 0 ,     
	\end{aligned}
	\right. 
\end{align} 

\noindent where $ M_i =  |y_i| / [  1+r  {\rm sign}(y_i) ]. $

The solutions of the likelihood equations (\ref{t+2}) provide the MLEs of $r, \alpha$, and $\eta$, which can be obtained using numerical procedures such as the Newton-Raphson type method. The implementation of the EM algorithm is summarized as follows:  
\begin{algorithm}[h]   
	\caption{ The EM  algorithm for solving the MLE of the SkeGTD. }    \label{algem}  
	\textbf{Initialization:} Set $t=0$ and given a set of initial values $ |r^0| \leq 1,~ \alpha^0 >0, ~ \eta^0 >0 $.          
	\begin{algorithmic}  
			\REPEAT      
			\STATE   1. Set $t=t+1$.  
			\STATE   2. Substitute $ \boldsymbol{\theta}^0 $ into the second equation of Equation (\ref{t+2}),  and we obtain  $\alpha^1$.    
			\STATE   3. Substitute $ \boldsymbol{\theta}^0  $ into the first and third equations of Equation (\ref{t+2})  and  
			minimize the sum of the  absolute values of the expressions on the left, obtaining $(r^1, \eta^1)$.      
			\UNTIL{ $ \mathop{ \max }\limits_{i} |\boldsymbol{\omega}_i^{t+1} - \boldsymbol{\omega}_i^{t} | \le   10^{-4} $ }.    
		\end{algorithmic}      
\end{algorithm}      
\vspace{-5mm}       
Generally, when employing numerical procedures, we recommend trying multiple initial values to search for the global optimum by comparing their log-likelihood values. The algorithm is iterated until the convergence criterion $\mathop{\max}\limits_{i} |\boldsymbol{\omega}_i^{t+1} - \boldsymbol{\omega}_i^{t}| \le 10^{-4}$ is met.

\subsection*{Existence and Uniqueness}  

We now investigate the existence and uniqueness of the MLE. For the parameter transformation, $\boldsymbol{\theta}_i = h_i(\boldsymbol{\omega}), i=1, 2, 3$, the information matrices of MLEs for $\boldsymbol{\omega}$ and $\boldsymbol{\theta}$, denoted by $J(\boldsymbol{\omega})$ and $I(\boldsymbol{\theta}$, satisfy   
\begin{equation}
	I^{-1}(\boldsymbol{\theta}) = \Delta (\boldsymbol{\omega})  J^{-1}( \boldsymbol{\omega})   \Delta^{T}  (\boldsymbol{\omega}), \label{IJ}
\end{equation} 
\noindent where $\Delta(\boldsymbol{\omega})$ is a diagonal matrix with entries $\delta_{ij}(\boldsymbol{\omega}) = \partial h_i(\boldsymbol{\omega})/\partial(\boldsymbol{\omega}_j)$, for $i, j = 1, 2, 3$. Specifically, $\Delta(\boldsymbol{\omega}) = {\rm diag}(1, 1, \boldsymbol{\omega}_3^{-2})$. We derive the closed-form expression for the information matrix $J(\boldsymbol{\omega})$. The nonzero elements of $J(\boldsymbol{\omega})$ are     
\begin{equation} \label{Jii}   
	\begin{split}
	J_{11} =& \    \frac{1}{1-r^2}   
	\left\{   1+  \frac{\alpha} { \alpha \eta + \eta ^2 + \eta}  \right\} , 
	 \\   
	J_{22} =& \    
	-  \psi ' (\alpha + \eta ) + \psi ' (\alpha) 
	-  \frac{ \eta ( \alpha + \eta  + 2 )}  { \alpha (\alpha + \eta ) ( \alpha  + \eta + 1 ) } ,    \\
	J_{33} =& \   
	-  \psi ' (\alpha + \eta ) + \psi ' (\eta) - 1/ \eta^2 
	+ \frac{2 \alpha }{\eta ( \alpha + \eta )}  
	\left\{  \psi (\eta +1 ) - \psi (\alpha) +  \log (2 \alpha)  \right\}    \\
	+& \   \frac{\alpha}{ \eta ^2 ( \alpha+\eta +1 ) } 
	\left\{   \psi (\eta +1 ) - \psi (\alpha+1) +  \log (2 \alpha) +
	\psi' ( \alpha +1 ) +  \psi' ( \eta +1 )  \right\}  ,    \\
	J_{23} =& \   
	\psi ' (\alpha + \eta ) - \frac{1} {\alpha + \eta }  
	+ \frac{ \psi (\eta +1 ) +  \log (2 \alpha)  }  { (\alpha + \eta ) ( \alpha  + \eta + 1 ) }  
	-  \  \frac{ \psi(\alpha )} { \alpha  + \eta }  
	+ \frac{\psi(\alpha +1) }  { \alpha  + \eta +1 }  .   
\end{split} 
\end{equation}         
Further details are provided in Appendix \ref{app.fisher}. From formula (\ref{IJ}), we can obtain the elements of $I(\boldsymbol{\theta})$. Additionally, numerical evidence indicates that the information matrix $I(\boldsymbol{\theta})$ remains positive definite when $\alpha,~\beta \in [0.5, 25] $.   

\subsection{ L-moments Estimation }     

To obtain more reliable estimates for small samples, we propose a method using L-moments, which is a reliable tool for inferences and is less affected by the presence of outliers because it assigns less weight to extreme data values (\citet{basalamah2018beta}).   L-moments exhibit strong robustness compared to conventional moments (\citet{hosking2006characterization}). Since \citet{hosking1990moments} introduced the method of L-moments to several distributions, it has been extended to various fields (\citet{li2014estimation, modarres2010regional, ouarda2016review}). L-moments, representing expectations of linear combinations of order statistics, are defined as 
\begin{align} 
	\lambda_k = \frac{1}{k}  \sum_{m=0}^{k-1}  
	(-1) ^ m \binom{k-1}{m} \ E(X_{k-m:k}),  \ \ k=1,2, 3, 4, 
\end{align} 	
\noindent where $ X_{1:k} \le X_{2:k} \le \cdots \le X_{k:k} $ are  the order statistics from a random sample of size $k$ drawn from the distribution of $ X$.  
It is now convenient to define the L-moments ratios as 
\begin{equation*}      
		\tau_3 = \lambda_3 / \lambda_2, ~~~ 
		\tau_4 = \lambda_4 / \lambda_2,  
\end{equation*} 
and the coefficient of L-variance (L-CV) is 
\begin{equation*}
	\tau_1 = \lambda_2 / \lambda_1,       
\end{equation*}  
where $\tau_3$ represents L-skewness, $\tau_4$ denotes L-kurtosis, and $\tau_1$ is a measure of scale or dispersion (\citet{hosking1990moments}).         
          
Now, let's consider using the L-moments method to estimate the parameters of the three-parameter SkeGTD. The first four L-moments are given as follows.       
\begin{mythm} \label{L1234}  
	Let $ X_0 \sim SkeGTD(0, 1, r, \alpha, \beta) $ with  cdf  $ F_{X_0}(x; r, \alpha, \beta )$,  then 	  
	\begin{align}            
		\lambda_1 = \   2&r  (2 \alpha)^ {1/ \beta  }   	B (\alpha- 1/ \beta , 2/ \beta  ) /  B(\alpha, 1/ \beta  ) \label{l1},  
		\\  
		\lambda_2  =\   (&r^2+1)  (2\alpha)^ { 1/ \beta } B(\alpha- 1/ \beta , 2 / \beta ) /  B(\alpha,1/ \beta )  \notag \\  
		-&  (3 r^2+1)  (2\alpha)^ { 1/ \beta  }  
		\   \sum_{k=0}^{\infty}  C_k  B\left(2\alpha+k-1/ \beta ,  2/ \beta  \right)  / [B(\alpha, 1/ \beta )]^2 ,  \label{l2}
		\\       
		\lambda_3    = \ 	2&r  (2 \alpha ) ^ {1/ \beta } 	B (\alpha- 1/ \beta , 2/ \beta ) /  B(\alpha, 1/ \beta )   \notag  \\  
		-& \   3 (r^3+3r)  (2\alpha)^  {1/ \beta }    
		\sum_{k=0}^{\infty}  C_k  B\left(2\alpha+k- 1/ \beta ,  2/ \beta  \right)  / [B(\alpha, 1/ \beta )]^2     	\label{l3}    \\  
		+& \   6 ( r^3+r)  (2\alpha)^ {1/ \beta }   
		\sum_{k=0}^{\infty}  D_k  B\left( 3\alpha+k- 1/ \beta ,  2 / \beta\right) /  [B(\alpha,  1/ \beta  )]^3 ,     
	   \notag  \\   
		\lambda_4     
		= \   \lambda_2&  - 5 h_{1}(r, \alpha, \beta) 
		\sum_{k=0}^{\infty}   D_k  B \left( 2\alpha+k- 1/ \beta ,  1/ \beta  \right)  / \left\{   4 \beta [B(\alpha, 1/ \beta )]^2 
		\right\}    \notag \\   
		-& \ 5 h_{2}(r, \alpha, \beta)  
		\sum_{k=0}^{\infty}   E_k  B\left( 3\alpha+k- 1/ \beta  , 1/ \beta  \right) / \left\{    4 \beta [B(\alpha,  1/ \beta   )]^3 
		\right\}       \label{l4}   \\   
		+& \  5  h_{3}(r, \alpha, \beta)   
		\sum_{k=0}^{\infty}   F_k  B\left(4\alpha+k-1/ \beta ,  1/ \beta  \right)  /  \left\{    \beta [2 B(\alpha, 1/ \beta  )]^4
		\right\},     \notag    
	\end{align}
	where   
	\begin{eqnarray}
		h_{i}(r, \alpha, \beta)  = \left\{ (1+r)^{i +2} + (1-r)^{i +2} \right\}  (2\alpha)^  {1/ \beta } \ \ ,  i=1, 2, 3,  \notag 
	\end{eqnarray}	
	\begin{eqnarray} 
		& C_k & = \binom{1/\beta -1}{k}  \frac{ (-1)^k }{\alpha + k} , \quad  \quad \quad
		D_k = \sum_{ i_1 + i_2 =k} C_{i_1}  C_{i_2}  , \notag \\
		& E_k & = \sum_{ i_1 + i_2  + i_3 =k} C_{i_1}  C_{i_2} C_{i_3} , \quad \quad ~  
		F_k = \sum_{ i_1 + i_2  + i_3 + i_4  =k} C_{i_1}  C_{i_2} C_{i_3} C_{i_4}. \notag 
	\end{eqnarray}		
\end{mythm} 
 The proof of Theorem \ref{L1234} is given in the Appendix \ref{app.l1234}.    

Let $x_{1:n} \le x_{2:n} \le \cdots \le x_{n:n}$ be the order statistics of a random sample $x_1, x_2, \cdots, x_n$ of size $n$. Then the first four sample L-moments are given as follows    
\begin{equation}  
	\begin{split}
		\ell_1 =& \frac{1}{n}  \sum_{i=1}^{n}  x_{i:n}, ~~~~~~~~~~~~~~~~~~~~~~~~~~~~ \ell_2 =   \frac{1}{n(n-1)}  \sum_{i=1}^{n}\  b_{i,n}   x_{i:n},   \\   
		\ell_3 =&  \frac{1}{n(n-1)(n-2)}  \sum_{i=1}^{n}\  c_{i,n}   x_{i:n} , ~~~~ 
		\ell_4 =   \frac{1}{n(n-1)(n-2)(n-3)}  \sum_{i=1}^{n}\  d_{i,n}   x_{i:n},	  
	\end{split}  
\end{equation} 
\noindent where $ b_{i,n}=2i-n-1,  c_{i,n}=(i-1)(i-2)-4(i-1)(n-i)+(n-i)(n-i-1),$ and $ d_{i,n}= (i-1)(i-2)(i-3)-(n-i)(n-i-1)(n-i-2)+9(n+1-2i)(i-1)(n-i). $    
The sample L-CV, sample L-skewness and sample L-kurtosis is defined as   
\begin{equation*}       
	  \hat{ \tau_1 } = \hat{ \lambda_2 } / \hat{ \lambda_1}, ~~\hat{ \tau_3 } = \hat{ \lambda_3 } / \hat{ \lambda_2}   ~~  \textrm{and} ~~ \hat{ \tau_4 } = \hat{ \lambda_4 } / \hat{ \lambda_2} ,
\end{equation*}         
respectively.     
Let $\tau_1$, $\tau_3$, and $\tau_4$ be equal to their corresponding estimators $\hat{\tau_1}$, $\hat{\tau_3}$,  and $\hat{\tau_4} $, respectively. Then we can obtain an equation set whose solution provides the estimation of $\theta$. This can be achieved through a numerical method such as the L-BFGS-B algorithm.

\subsection{ Parameter Estimation of $ SkeGTD(\mu, \sigma, r, \alpha, \beta) (\mu \neq  0, \sigma \neq 1) $   }   \label{twostep}  

In this subsection, we propose a two-step estimation approach to estimate the parameters of $ SkeGTD(\mu, \sigma, r, \alpha, \beta) $.

In the first stage, we generate a set of $\mu$ values within a neighborhood of $\hat{\mu}_0$. For each given $\mu$, the other four parameters are estimated separately. In the second stage, a numerical interpolation method using the corresponding likelihood values is employed to determine a curve, defined as the profile likelihood for $\mu$. The maximum of the profile likelihood serves as the estimator for $\mu$. Then $\sigma, r, \alpha, $ and $ \beta $ are re-estimated and taken as the final estimates.

Let $ X_{1},  \cdots ,X_{n}  $ be a random sample  from $  SkeGTD(\mu, \sigma, r, \alpha, \beta) $.      
TSE  procedures are outlined below.       

\textbf{1.}   Initialize $\hat{\mu}_0$  and  generate a set of  $   \mu_{(t)}, t=1, \dots, s $ at an    appropriate distance  in a neighborhood of $\hat{\mu}_0$. 

\textbf{2.}  For each fixed $ \mu_{(t)} $, estimate $r$ as follows:  
\begin{equation}   
	 r_{(t)} = 1 - \frac{2}{n}  \sum_{i=1}^{n} I (X_i \leq \mu_{(t)})   ,   \notag   
\end{equation}  
where $I(\cdot)$ denotes the indicator function for $X_i$.

\textbf{3.}  Estimate  $  \alpha, \beta $ and $ \sigma $ using the method of moments.      
    
Replacing $r$ in Equations (\ref{g1}) and (\ref{g2}) with $r_{(t)}$, we obtain the estimators of $\alpha$ and $\beta$ as:   
\begin{equation}       
	(\alpha_{(t)}, \beta_{(t)})  =  \mathop{\arg \min }\limits_{ \alpha,~ \beta } \  \left\{  	|\gamma_1(X;  r_{(t)},\alpha,\beta)-g_1| +  |\gamma_2(X;  r_{(t)},\alpha,\beta)-g_2| \right\},     
\end{equation}    
\noindent where $ g_1 $ and $ g_2 $ the sample skewness and kurtosis, respectively. Replacing $\alpha$ and $\beta$ in Equation (\ref{sgtvar}) with $\alpha_{(t)}$ and $\beta_{(t)}$, we have  
\begin{equation}       
	\sigma_{(t)}  =  \mathop{\arg \min }\limits_{ \sigma } \    	|Var(X;  \sigma, r_{(t)}, \alpha_{(t)}, \beta_{(t)})-s^2|,       
\end{equation}
\noindent where $ s^2 $ denotes the sample variance. 
     
\textbf{4.} Compute the likelihood function for $t = 1, 2, \dots, s$. Then, interpolate on $\mu$ values to obtain the profile likelihood, with its maximizer chosen as the estimator of $\mu $. Repeat steps 2 and 3 to obtain the final estimates of $r, \alpha, \beta, and \sigma$.                 

It is essential to choose a suitable starting value, $\hat{\mu}_0$, in step 1. Since $\mu$ represents the mode of the SkeGTD, we recommend using the robust estimator, $\hat{\mu}_0$, obtained through the half-range mode (HRM) method (\citet{bickel2002robust}).   The HRM estimator can be computed using the function $  half.range.mode $  in the $ genefilter $ package or the function $ hrm $  in the  $ modeest $ package in R software. The algorithm proceeds as follows.

\begin{compactitem}   
	\item[1.]  Let $ S_0= { y_i ^{(0)}, i=1, \cdots, n } $ denote the observed sample data, where $ y_{(1)}^{(0)} \leq \cdots \leq y_{(n)}^{(0)} $ are the order statistics, and $ w_1 = (y_{(n)}^{(0)} - y_{(1)}^{(0)}) /2$. 
	\item[2.]   Determine the target interval $ [M_1(w_1) - w_1 / 2, M_1(w_1) + w_1 / 2] $ that includes the maximum possible number of points and the midpoint $ M_1(w_1) $.  
	\item[3.]   The points within this interval constitute the subset we choose, denoted as $ S_1 = { y_i ^{(1)}, i=1, \cdots, n_1 }$.     
	\item[4.]  Repeat steps 1, 2, and 3, selecting a subset from $S_1$, and continue this process. The iteration stops when only two points remain. The estimator of the mode is the average of these two sample points.     
\end{compactitem}   

\section{ Simulation Study }   \label{sec4}

\subsection{Experiment I:  performance of the estimation methods }    
                 
To investigate the finite-sample performance of different estimation methods, we conducted a simulation study using sample sizes of $n = 20, 50, 100, 200,$ and $500$. Assuming $\mu = 0$ and $\sigma = 1$, we considered two sets of true values for $\boldsymbol{\theta}$: $(-0.5, 4, 2.5)^\top$ and $(0.7, 3, 2.5)^\top$. For each set of true values of $\boldsymbol{\theta} $ and every sample size $n$, we generated $N = 5000$ random samples from the SkeGTD based on formula (\ref{structure}). Estimators $ \theta_{(i)}$ (where $i = 1, 2, \dots, N$) were obtained using different methods. We calculated the average, relative bias (Rbias), and mean squared error (MSE) of these estimators for each parameter. The Rbias and MSE are defined as follows       
\[  
Rbias (\hat{\theta}) = \frac{1}{N} \sum_{i=1}^{N} 
\left| \frac{\hat{\theta}_{(i)} - \theta}  {\theta} \right|, ~~  \textrm{and} ~~	MSE(\hat{\theta}) =   \frac{1}{N} \sum_{i=1}^{N} 
(\hat{\theta}_{(i)} - \theta)^2  .        
\]

Tables \ref{3can.2}-\ref{3can.5} present the simulation results for the MLE and LME of $r, \alpha,$ and $\beta$. To describe the Rbias and MSE of the MLEs and LMEs, we graphically present the simulation results with $\boldsymbol{\theta} = (0.7, 3, 2.5)$ in Figure \ref{bar}.    As shown in Figure \ref{bar}, for $r, \alpha$, and $\beta$, the Rbias and MSE for each method decrease as $n$ increases. The Rbias of $\hat{r}$ is the smallest among the three parameters. It is noteworthy that the shape of the SkeGTD is more sensitive to changes in $\alpha$ when $\alpha$ is small. This implies that it might be beneficial to impose some restrictions on the value of $\alpha$ to avoid the difficulty of estimating large $\alpha$.  Tables \ref{3can.2}-\ref{3can.5} show that the MLE $\hat{\beta}$ tends to overestimate $\beta$, while the LME $\hat{\beta}$ tends to underestimate $\beta$. Moreover, it is evident from Tables \ref{3can.2}-\ref{3can.5} that for $\alpha$, LME outperforms the MLEs in terms of MSE for small sample sizes. However, for $r$ and $\beta$, the MLEs are more accurate than the LMEs across all sample sizes. In summary, the simulation studies indicate that LME is more accurate than MLE for small and moderate sample sizes.
 
Next, we conducted a simulation study of the TSE for the $SkeGTD(\mu, \sigma, r, \alpha, \beta)$. The true values of the parameter vector $(\mu, \sigma, \theta)^\top$ were set to $(4, 2, 0.7, 4, 2)^\top$ and $(-4, 2, 0.7, 4, 2)^\top$. The corresponding simulation results are presented in Tables \ref{5can.3}-\ref{5can.4}.  

Tables \ref{5can.3}-\ref{5can.4} reveal that $\hat{\sigma}, \hat{\alpha}$, and $\hat{\beta}$ become more accurate as $n$ increases. The MSEs of $\hat{\mu}$ and $\hat{r}$ are smaller than those of the other parameters. Furthermore, we observed that the estimators of $\sigma, \alpha$, and $\beta$ obtained by the two-step estimation method exhibit good performance.                             

\subsection{Experiment II: model selection}        

In this simulation study, we aim to demonstrate the flexibility of our SkeGTD model. As suggested by an anonymous reviewer, we generated artificial data from a distribution distinct from SkeGTD but exhibiting skewness and heavy-tailed properties. A suitable example is the skew-cauchy (SC) distribution, with a density function of     
\begin{equation}
	 f_{SC} (x; \xi, \omega, \alpha)=\frac{1}{\pi \omega\left(1+\left(\frac{x-\xi}{\omega}\right)^2\right)}   
	 	\left(  1+ \frac{\alpha \left(\frac{x-\xi}{\omega}\right) }{ \sqrt{ 1 + \left(\frac{x-\xi}{\omega}\right)^2 (1+\alpha^2)} }  \right) , ~~ x \in \mathbb{R},          
\end{equation}     
where $ \xi, \alpha \in \mathbb{R}$, and $ \omega >0 $.        

In this study, we generated 1000 samples from the SC distribution, with sample sizes of $n=20, 50, 100, 200,$ and 500, respectively. The parameters are set as $ \xi = -1.8, \omega = 0.8, \text{~and~} \alpha = 18 $, which correspond to a highly right-skewed and heavy-tailed distribution.   For each sample, we employed the normal (N), $t$, SN (\cite{azzalini1985class}), skew-t (ST; \cite{fernandez1998bayesian}) distributions, and SkeGTD for fitting. The pdf of ST is given by               
\begin{eqnarray}     
	\begin{aligned} 
	 f_{ST}(x; \nu, \gamma) =
		\begin{dcases}  
			\frac{2 }{ \gamma + \frac{1}{\gamma} } f_{t}(\gamma x; \nu) , & x \le  0, \\ 
			\frac{2 }{ \gamma + \frac{1}{\gamma} } f_{t}\left( \frac{x}{\gamma}; \nu \right) , & x > 0, 
		\end{dcases}	 	   
	\end{aligned} 
\end{eqnarray}            
where $ \nu >0, \gamma \in \mathbb{R} $, $ f_{t}(\cdot; \nu) $ denotes the pdf of the $t$ distribution with $ \nu $ degrees of freedom. We compared the SkeGTD with four alternative models using the Akaike Information Criterion (AIC; \cite{akaike1998information}), the Bayesian Information Criterion (BIC; \cite{schwarz1978estimating}), and the Efficient Determination Criterion (EDC; \cite{bai1989rates}), such that $ AIC = -2 \ell (\widehat{\theta}) + 2 \rho,  BIC= -2 \ell (\widehat{\theta}) + \log(n) \rho, EDC = -2 \ell (\widehat{\theta}) + 0.2 \sqrt n \rho $, where $ \ell (\widehat{\theta}) $ is the actual log-likelihood, $\rho $ is the number of free model parameters to be estimated (\cite{lye1993robust}). Table \ref{tablesimu2} reports the percentage in which the SkeGTD was indicated as the best model according to the three criteria.       
        
The results shown in Table \ref{tablesimu2} suggest that the SkeGTD's percentage of preference increases with increasing sample size $n$. For large sample sizes $n= 200$ and 500, SkeGTD shows a significantly higher percentage of preference when compared to the N and SN models than to the $t$ and ST models. This observation indicates the superior performance of SkeGTD in fitting skewed and heavy-tailed data. Even compared to the ST distribution, SkeGTD maintains a superior win rate of 60\% to 70\% for all $n$ considered. In summary, all model selection criteria demonstrate the superior fit of the SkeGTD for the simulated data, thus signifying the indispensability of the proposed model.

\section{Real Data Analysis }   \label{sec5}  
         
In this section, we provide two applications for the proposed SkeGTD. The first application involves fitting a real dataset and comparing the results with some alternative models. The second application demonstrates how to incorporate the SkeGTD into a regression model.

The first example considers the roller data set which is available in   \url{   http://lib.stat.cmu.edu/jasadata/laslett}.   This dataset consists of 1,150 height measurements at 1µm intervals along the drum of a roller and was collected as part of a study on surface roughness.  Basic descriptive statistics for the dataset are summarized in Table \ref{summary roller}.   

Notably, the sample skewness and kurtosis coefficients indicate that the data exhibits moderate to strong asymmetry and heavy tails. This suggests that fitting a skewed distribution to the data is a reasonable choice. Using the TSE method discussed in Section \ref{twostep}, we obtained the following estimators, with bootstrap standard errors in parentheses: $\hat{\mu} = 3.819(0.0227), \hat{\sigma} = 0.445(0.0354), \hat{r} = -0.300(0.0210), \hat{\alpha} = 9.000(2.1092), $ and $ \hat{\beta} =  1.572(0.1460) $.  The standard errors indicate that $\mu, \sigma, r,$ and $\beta$ are reliably estimated, while the estimator of $\alpha$ exhibits some variability. The values of $\hat{r}, \hat{\alpha},$ and $\hat{\beta}$ suggest that the data exhibits asymmetry and heavy tails. The bootstrap 95\% confidence intervals (BCI) for these estimators are provided in Table \ref{aic}. Notably, the interval corresponding to the skewness parameter $r$ does not contain zero, suggesting evidence to reject the hypothesis that this dataset follows a normal distribution.

This dataset has previously been modeled using the skew-normal distribution, skew-flexible-normal (SFN) distribution, and the alpha-skew generalized $t$ (ASGT) distribution proposed by \citet{azzalini1985class}, \citet{gomez2011bimodal}, and \citet{acitas2015alpha}, respectively. Figure \ref{roller} shows the fitting results superimposed on a single set of coordinate axes. From the histogram, it is evident that the data exhibit a clear right-tailed distribution, and the SkeGTD has superiority for modeling the peakedness of the roller data. Thus, the SkeGTD provides the best fitting among these alternatives.   
Meanwhile, we compared the SkeGTD with other skewed distributions using the log-likelihood (logL) and the AIC criterion. Detailed results are presented in Table \ref{aic}. It's worth noting that lower AIC values and higher logL values indicate better model fitting. Both criteria suggest that the proposed SkeGTD performs satisfactorily.

We now present the second example in which the SkeGTD is applied in estimating a regression model. The Marlin Marietta dataset used in this example is sourced from \citet{butler1990robust}, where a simple linear regression model is introduced:         
\begin{equation}
	y_i =  \beta_{0} + \beta_{1} x_i + \epsilon_i,~~ i= 1,\dots, 60, \label{reg1} 
\end{equation}   
where $y_i$ represents the excess rate of return for the Martin Marietta Company observed from January 1982 to December 1986, $x_i$ denotes the Center for Research in Security Prices (CRSP), an index of the excess rate of return for the New York market, and $\epsilon_i$ is an error term assumed to be independently and identically distributed (i.i.d) with a generalized $t$ distribution. From Figure \ref{MMdata} (a), we can observe that the data include one obvious extreme observation in the upper right corner and two potential outliers determined by residuals and other diagnostic measures. Figure \ref{MMdata} (b) indicates that the data exhibit high skewness and heavy tails. Previous studies by \citet{azzalini2003distributions} and \citet{diciccio2004inferential} involved fitting different skewed distributions to this dataset. In their papers, the error term is assumed to follow a skewed $t$ distribution and skewed exponential power (SEP) distribution.  
 
In this study, we assume that $\epsilon_i$ follows a SkeGTD with $\mu = 0$, which is equivalent to $y_i|x_i \sim SkeGTD(\beta_{0} + \beta_{1} x_i, \sigma, r, \alpha, \beta)$. The MLEs and their corresponding standard errors, calculated via the empirical information matrix, are as follows: $\hat{\beta}{0} = 0.004$, $\hat{\beta}{1} = 1.112$, $\hat{\sigma} = 0.061$, $\hat{r} = 0.250 (0.087)$, $\hat{\alpha} = 0.761 (0.201)$, and $\hat{\beta} = 3.174 (0.377)$.  It's important to note that a small value of the tail parameter $\alpha$ implies heavy tails in the data. The above estimators indicate that the null hypotheses $ H_0: r = 0$ and $H_0: \alpha > 1$ are rejected at the 5\% level. This suggests that the data exhibit asymmetry and heavy tails. 

Figure \ref{MMdata} (a) shows a scatter plot with three regression lines superimposed. The dashed line corresponds to the ordinary least squares (OLS) fit, while the dotted line and the solid line correspond to the SEP distribution and SkeGTD, respectively. It's evident that the fitted models obtained by the SEP and SkeGTD are hardly influenced by outliers. To account for the asymmetry in the distribution of $\epsilon$, the intercept can be adjusted to $\hat{\beta_{0}} + \hat{\beta_{1}} x + \hat{ {E}}(\epsilon)$. Figure \ref{MMdata} (b) shows the histogram of the residuals after removing the line $  0.004+1.112 CRSP +  \hat{ {E}}(\epsilon) $,  which further supports the satisfactory agreement between the original data and the fitted SkeGTD.

\section{ Conclusion and Discussion}  \label{sec6}  

We have proposed a novel class of skewed generalized $ t $ distribution as a scale mixture of the SGN distribution.  The SkeGTD has broad applicability in statistical modeling due to its capacity to accommodate a wide range of skewness and kurtosis, and its flexibility in handling separate location, scale, skewness, and shape parameters. Moreover, it contains several common distributions as special cases, including the Normal, Laplace, Uniform, Student-t, Generalized Normal distributions, and more. 
    
In this paper, we have presented the stochastic representation of the SkeGTD and several key results. For the three-parameter SkeGTD, we have developed the likelihood inference method based on the stochastic representation and  L-moments method. Additionally, we have demonstrated that the information matrix remains positive definite when $ \alpha,\beta \in [0.5, 25] $. Furthermore, we have introduced a TSE method for the five-parameter SkeGTD. The simulation results have shown that the LMEs of $\alpha$ and $\beta$ are more accurate than the MLEs for small or medium sample sizes, while the MLEs of $r$ and $\beta$ are more accurate for large sample sizes. The proposed SkeGTD has proven to be the best-performing distribution among the existing ones in the analysis of the simulated data and two real datasets.             

Finally, we have only considered situations where the data is unimodal. There is no doubt that the concept of the SkeGTD can also be extended to multimodal data, and we are currently exploring mixture modeling based on independent skewed generalized $t$ distributions. The proposed SkeGTD can be further extended for random effects modeling in linear mixed models and latent variable mixture modeling for measurement error models. We are also planning to apply it to various domains such as image processing, survival analysis, and more.

\section*{Acknowledgments}             

The authors wish to express their deepest gratitude to the chief editor, the associate editor, and two referees whose careful reading and comments led to an improved version of the paper.  This study was supported by the National Natural Science Foundation of China(No.11701021),  National Statistical Science Research Project (No.2022LZ22), and Science and Technology Program of Beijing Education Commission (No.KM202110005013).

\section*{Conflict of interest}       

No potential conflict of interest was reported by the authors.  

\vskip 3mm 
 
\bibliographystyle{apalike}       
\bibliography{Myref}

\begin{thebibliography}{}

\bibitem[Aas and Haff, 2006]{aas2006generalized}
Aas, K. and Haff, I.~H. (2006).
\newblock The generalized hyperbolic skew student’st-distribution.
\newblock {\em Journal of financial econometrics}, 4(2):275--309.

\bibitem[Acitas et~al., 2015]{acitas2015alpha}
Acitas, S., Senoglu, B., and Arslan, O. (2015).
\newblock Alpha-skew generalized t distribution.
\newblock {\em Revista Colombiana de Estad{\'Y}stica}, 38(2):353--370.

\bibitem[Akaike, 1998]{akaike1998information}
Akaike, H. (1998).
\newblock Information theory and an extension of the maximum likelihood
  principle.
\newblock In {\em In Selected papers of hirotugu akaike}, pages 199--213.
  Springer.

\bibitem[Arslan and Genc, 2009]{arslan2009skew}
Arslan, O. and Genc, A.~I. (2009).
\newblock The skew generalized t distribution as the scale mixture of a skew
  exponential power distribution and its applications in robust estimation.
\newblock {\em Statistics}, 43(5):481--498.

\bibitem[Asgharzadeh et~al., 2013]{asgharzadeh2013approximate}
Asgharzadeh, A., Esmaily, L., and Nadarajah, S. (2013).
\newblock Approximate mles for the location and scale parameters of the skew
  logistic distribution.
\newblock {\em Statistical Papers}, 54:391--411.

\bibitem[Azzalini, 1985]{azzalini1985class}
Azzalini, A. (1985).
\newblock A class of distributions which includes the normal ones.
\newblock {\em Scandinavian Journal of Statistics}, 12:171--178.

\bibitem[Azzalini and Capitanio, 2003]{azzalini2003distributions}
Azzalini, A. and Capitanio, A. (2003).
\newblock Distributions generated by perturbation of symmetry with emphasis on
  a multivariate skew t-distribution.
\newblock {\em Journal of the Royal Statistical Society: Series B (Statistical
  Methodology)}, 65(2):367--389.

\bibitem[Azzalini and Dalla~Valle, 1996]{azzalini1996multivariate}
Azzalini, A. and Dalla~Valle, A. (1996).
\newblock The multivariate skew-normal distribution.
\newblock {\em Biometrika}, 83(4):715--726.

\bibitem[Azzalini and Genton, 2008]{azzalini2008robust}
Azzalini, A. and Genton, M.~G. (2008).
\newblock Robust likelihood methods based on the skew-t and related
  distributions.
\newblock {\em International Statistical Review}, 76(1):106--129.

\bibitem[Bai et~al., 1989]{bai1989rates}
Bai, Z.-D., Krishnaiah, P.~R., and Zhao, L.-C. (1989).
\newblock On rates of convergence of efficient detection criteria in signal
  processing with white noise.
\newblock {\em IEEE Transactions on Information Theory}, 35(2):380--388.

\bibitem[Balakrishnan and Scarpa, 2012]{balakrishnan2012multivariate}
Balakrishnan, N. and Scarpa, B. (2012).
\newblock Multivariate measures of skewness for the skew-normal distribution.
\newblock {\em Journal of Multivariate Analysis}, 104(1):73--87.

\bibitem[Basalamah et~al., 2018]{basalamah2018beta}
Basalamah, D., Ning, W., and Gupta, A. (2018).
\newblock The beta skew t distribution and its properties.
\newblock {\em Journal of Statistical Theory and Practice}, 12(4):837--860.

\bibitem[Bickel, 2002]{bickel2002robust}
Bickel, D.~R. (2002).
\newblock Robust estimators of the mode and skewness of continuous data.
\newblock {\em Computational statistics \& data analysis}, 39(2):153--163.

\bibitem[Butler et~al., 1990]{butler1990robust}
Butler, R.~J., McDonald, J.~B., Nelson, R.~D., and White, S.~B. (1990).
\newblock Robust and partially adaptive estimation of regression models.
\newblock {\em The review of economics and statistics}, 72:321--327.

\bibitem[Chen et~al., 1999]{chen1999new}
Chen, M.~H., Dey, D.~K., and Shao, Q.~M. (1999).
\newblock A new skewed link model for dichotomous quantal response data.
\newblock {\em Journal of the American Statistical Association},
  94(448):1172--1186.

\bibitem[DiCiccio and Monti, 2004]{diciccio2004inferential}
DiCiccio, T.~J. and Monti, A.~C. (2004).
\newblock Inferential aspects of the skew exponential power distribution.
\newblock {\em Journal of the American Statistical Association},
  99(466):439--450.

\bibitem[Fern{\'a}ndez and Steel, 1998]{fernandez1998bayesian}
Fern{\'a}ndez, C. and Steel, M.~F. (1998).
\newblock On bayesian modeling of fat tails and skewness.
\newblock {\em Journal of the american statistical association},
  93(441):359--371.

\bibitem[Ferreira and Steel, 2006]{ferreira2006constructive}
Ferreira, J. T.~S. and Steel, M. F.~J. (2006).
\newblock A constructive representation of univariate skewed distributions.
\newblock {\em Journal of the American Statistical Association},
  101(474):823--829.

\bibitem[Garc{\'\i}a et~al., 2010]{garcia2010new}
Garc{\'\i}a, V.~J., G{\'o}mez-D{\'e}niz, E., and V{\'a}zquez-Polo, F.~J.
  (2010).
\newblock A new skew generalization of the normal distribution: Properties and
  applications.
\newblock {\em Computational statistics \& data analysis}, 54(8):2021--2034.

\bibitem[G{\'o}mez et~al., 2011]{gomez2011bimodal}
G{\'o}mez, H.~W., Elal-Olivero, D., Salinas, H.~S., and Bolfarine, H. (2011).
\newblock Bimodal extension based on the skew-normal distribution with
  application to pollen data.
\newblock {\em Environmetrics}, 22(1):50--62.

\bibitem[Guan et~al., 2021]{math9192413}
Guan, R., Zhao, X., Cheng, W., and Rong, Y. (2021).
\newblock A new generalized t distribution based on a distribution construction
  method.
\newblock {\em Mathematics}, 9(19):2413.

\bibitem[Hansen et~al., 2006]{hansen2006partially}
Hansen, J.~V., McDonald, J.~B., and Turley, R.~S. (2006).
\newblock Partially adaptive robust estimation of regression models and
  applications.
\newblock {\em European journal of operational research}, 170(1):132--143.

\bibitem[Hosking, 2006]{hosking2006characterization}
Hosking, J. (2006).
\newblock On the characterization of distributions by their l-moments.
\newblock {\em Journal of Statistical Planning and Inference}, 136(1):193--198.

\bibitem[Hosking, 1990]{hosking1990moments}
Hosking, J.~R. (1990).
\newblock L-moments: Analysis and estimation of distributions using linear
  combinations of order statistics.
\newblock {\em Journal of the Royal Statistical Society: Series B
  (Methodological)}, 52(1):105--124.

\bibitem[Lee and McLachlan, 2014]{lee2014finite}
Lee, S. and McLachlan, G.~J. (2014).
\newblock Finite mixtures of multivariate skew t-distributions: some recent and
  new results.
\newblock {\em Statistics and Computing}, 24(2):181--202.

\bibitem[Li et~al., 2014]{li2014estimation}
Li, X., Zuo, Y., Zhuang, X., and Zhu, H. (2014).
\newblock Estimation of fracture trace length distributions using probability
  weighted moments and l-moments.
\newblock {\em Engineering geology}, 168:69--85.

\bibitem[Lin, 2010]{lin2010robust}
Lin, T.~I. (2010).
\newblock Robust mixture modeling using multivariate skew t distributions.
\newblock {\em Statistics and Computing}, 20(3):343--356.

\bibitem[Lin et~al., 2007]{lin2007robust}
Lin, T.~I., Lee, J.~C., and Hsieh, W.~J. (2007).
\newblock Robust mixture modeling using the skew t distribution.
\newblock {\em Statistics and computing}, 17(2):81--92.

\bibitem[Lye and Martin, 1993]{lye1993robust}
Lye, J.~N. and Martin, V.~L. (1993).
\newblock Robust estimation, nonnormalities, and generalized exponential
  distributions.
\newblock {\em Journal of the American Statistical Association},
  88(421):261--267.

\bibitem[Modarres, 2010]{modarres2010regional}
Modarres, R. (2010).
\newblock Regional dry spells frequency analysis by l-moment and multivariate
  analysis.
\newblock {\em Water resources management}, 24(10):2365--2380.

\bibitem[Ouarda et~al., 2016]{ouarda2016review}
Ouarda, T.~B., Charron, C., and Chebana, F. (2016).
\newblock Review of criteria for the selection of probability distributions for
  wind speed data and introduction of the moment and l-moment ratio diagram
  methods, with a case study.
\newblock {\em Energy Conversion and Management}, 124:247--265.

\bibitem[Schroth and Muma, 2021]{9465733}
Schroth, C.~A. and Muma, M. (2021).
\newblock Real elliptically skewed distributions and their application to
  robust cluster analysis.
\newblock {\em IEEE Transactions on Signal Processing}, 69:3947--3962.

\bibitem[Schwarz, 1978]{schwarz1978estimating}
Schwarz, G. (1978).
\newblock Estimating the dimension of a model.
\newblock {\em The annals of statistics}, pages 461--464.

\bibitem[Theodossiou, 1998]{theodossiou1998financial}
Theodossiou, P. (1998).
\newblock Financial data and the skewed generalized t distribution.
\newblock {\em Management Science}, 44(12):1650--1661.

\bibitem[Venegas et~al., 2012]{venegas2012robust}
Venegas, O., Rodriguez, F., Gomez, H.~W., Olivares-Pacheco, J.~F., and
  Bolfarine, H. (2012).
\newblock Robust modeling using the generalized epsilon-skew-t distribution.
\newblock {\em Journal of Applied Statistics}, 39(12):2685--2698.

\bibitem[Zhu, 2012]{zhu2012asymmetric}
Zhu, D. (2012).
\newblock Asymmetric parametric distributions and a new class of asymmetric
  generalized t-distribution.
\newblock {\em Working paper}, Accessed May 1, 2016.
  \href{http://ssrn.com/abstract=2427545}{http://ssrn.com/abstract=2427545}.

\bibitem[Zhu and Galbraith, 2010]{zhu2010generalized}
Zhu, D. and Galbraith, J.~W. (2010).
\newblock A generalized asymmetric student-t distribution with application to
  financial econometrics.
\newblock {\em Journal of Econometrics}, 157(2):297--305.

\end{thebibliography}

\clearpage         

 \appendix   

 \section*{ Appendix }   
\setcounter{equation}{0}
\setcounter{subsection}{0}
\renewcommand{\theequation}{A.\arabic{equation}}
\renewcommand{\thesubsection}{A.\arabic{subsection}}

\subsection{  Calculation of the elements of the information matrix \label{app.fisher}  }

It is well known that the  entries  of $   J(\boldsymbol{\omega}) $ are given by   
\begin{equation*}    
   J_{ij}  =     -E_{\boldsymbol{\theta}}  \left( \frac{ \partial ^2 \log f(x; \boldsymbol{\omega}) } 
  { \partial \boldsymbol{\omega}_i \partial  \boldsymbol{\omega}_j } \right) ,\ \   i, j =1,2,3,       
\end{equation*}              
where 
$ (\boldsymbol{\omega}_1, \boldsymbol{\omega}_2, \boldsymbol{\omega}_3)^T = (r, \alpha, -1/ \beta)^T $.

The components of the score vector   $ \boldsymbol{S_\omega}   $  are 
\begin{align}    
	S_r =& \  
	 \frac{( \alpha \beta + 1) {\rm sign}(x) } { 1 + r {\rm sign}(x)} 
	\left( \frac{M_{x, \boldsymbol{\omega}}-1} {M_{x, \boldsymbol{\omega}}} \right),  \notag \\ 
	S_\alpha  =& \ 
	 \psi ( \alpha + \eta ) - \psi ( \alpha )
	- \eta / \alpha -\log M_{x, \boldsymbol{\omega}} + \frac{ \alpha + \eta } {\alpha} 
	\left( \frac{M_{x, \boldsymbol{\omega}}-1} {M_{x, \boldsymbol{\omega}}} \right),  \notag \\
	S_\eta  =& \ 
	 \psi ( \alpha + \eta ) - \psi (\eta) 
	- 1/ \eta - \log (2 \alpha ) - \log M_{x, \boldsymbol{\omega}}  
	+  \log ( A_{x, r} ) \frac{ \alpha + \eta } {\eta ^2}
	\left( \frac{M_{x, \boldsymbol{\omega}}-1} {M_{x, \boldsymbol{\omega}}} \right)   , \notag      
\end{align} 
where  
\begin{align*}
	 M_{x, \boldsymbol{\omega}}  =&  1+  \frac{ |x|^{\beta}  }{    2\alpha[ 1+r {\rm sign}(x) ]^ { \beta}   },   \notag \\ 
	 A_{x, r} =& \frac{|x|} { 1+r {\rm sign}(x)   } 
\end{align*}  
and $ \beta = -1 / \eta $.      
With some straightforward algebra manipulations,  we can obtain the elements of this matrix       
 \begin{equation}  
	\begin{aligned}  
		J_{11} =& \   \label{Jyuanshi}
		\frac{\alpha + \eta} {\eta} \kappa_3
		+ \frac{ \alpha + \eta  } {\eta^2}  \kappa_4 ,  \\
		J_{22} =& \ 
		\psi ' (\alpha)  - \psi ' (\alpha + \eta )  - \frac{\eta} {\alpha ^2}  
		- \frac{\alpha - \eta}{\alpha^2 } \kappa_1
		+ \frac{\alpha + \eta}{\alpha ^2 } \kappa_2 ,  \\
		J_{33} =& \  
		\psi ' (\eta)  - \psi ' (\alpha + \eta )  - \frac{1} {\eta^2 } 
		+ \frac{2 \alpha }{\eta ^3}  \kappa_5 
		\  +  \frac{\alpha + \eta} {\eta^4}  \kappa_7 , \\
		J_{12} =& \   
		- \frac{1} {\eta}  \kappa_8 
		+ \frac{\alpha + \eta } { \alpha \eta } \kappa_9
		  , \\
		J_{13} =& \    
		- \frac{ \alpha } { \eta^2 } \kappa_8 
		-  \frac{ \alpha + \eta } { \eta^3 } \kappa_{10} ,   \\
		J_{23} =& \   
		\psi ' (\alpha + \eta ) - \frac{1} {\alpha}  
		+ \frac{1} {\alpha} \kappa_1  
		+ \frac{1} {\eta ^2 }    \kappa_5
		\  -   \frac{\alpha + \eta} {\alpha \eta^2}  \kappa_6  ,    
	\end{aligned} 
\end{equation} 
\noindent where   
\begin{align}
	\kappa_1 = & \ \notag
	E\left[ \frac{M_{x, \boldsymbol{\omega}}-1} {M_{x, \boldsymbol{\omega}}} \right]  
	= \  \frac{\eta}{\alpha + / \eta} , \\
	\kappa_2 = & \  \notag
	E\left[ \frac{M_{x, \boldsymbol{\omega}}-1} {M_{x, \boldsymbol{\omega}}^2} \right]  
	= \ 
	\frac{ B(\alpha+1, \eta +1) }{B(\alpha, \eta) }  , \\
	\kappa_3 = & \  \notag
	E\left[ \frac{M_{x, \boldsymbol{\omega}}-1} {M_{x, \boldsymbol{\omega}}} \frac{1} { [ 1+r {\rm sign}(x) ] ^2} \right]  
	= \  \frac{1}{1-r^2}  \frac{\eta}{\alpha + / \eta}   , \\
	\kappa_4 = & \ \notag
	E\left[ \frac{M_{x, \boldsymbol{\omega}}-1} {M_{x, \boldsymbol{\omega}}^2} \frac{1} { [ 1+r {\rm sign}(x) ] ^2} \right]  
	= \  \frac{1}{1-r^2}  \frac{ B(\alpha+1, \eta +1) }{B(\alpha, \eta) },  \\
	\kappa_5 = & \  \notag
	E\left[ \log ( A_r ) \frac{M_{x, \boldsymbol{\omega}}-1}{M_{x, \boldsymbol{\omega}}}   \right]  
	= \  \frac{\eta^2}{\alpha + / \eta} 
	\left[  \psi (\eta +1 ) - \psi (\alpha) +  \log (2 \alpha)  \right] , \\
	\kappa_6 = & \  \notag
	E\left[ \log ( A_r ) \frac{M_{x, \boldsymbol{\omega}}-1}{M_{x, \boldsymbol{\omega}}^2}   \right]  
	= \   \frac{ \eta  B(\alpha+1, \eta +1) }{B(\alpha, \eta) } 
	\left[  \psi ( \eta +1 ) - \psi ( \alpha +1 ) +  \log (2 \alpha)  \right] , \\    
	\kappa_7 = & \   E\left[ \left[ \log ( A_r ) \right] ^2 \frac{M_{x, \boldsymbol{\omega}}-1}{M_{x, \boldsymbol{\omega}}^2}   \right]   
	=  \kappa_6 +  \frac{ \eta  B(\alpha+1, \eta +1) }{B(\alpha, \eta) } 
	\left[  \psi' ( \alpha +1 ) +  \psi' ( \eta +1 )  \right]  , \notag \\ 
	\kappa_8 = & \  E\left[ \frac{{\rm sign}(x)} {1+ r {\rm sign}(x)}  \frac{M_{x, \boldsymbol{\omega}}-1}{M_{x, \boldsymbol{\omega}}}   \right] = 0   , \ \ \ \ \ \ \ \ 
	\kappa_9 =   E\left[ \frac{{\rm sign}(x)} {1+ r {\rm sign}(x)}  \frac{M_{x, \boldsymbol{\omega}}-1}{M_{x, \boldsymbol{\omega}}^2}   \right] = 0  \notag , \\ 
	\kappa_{10} = & \  E\left[ \log ( A_r ) \frac{{\rm sign}(x)} {1+ r {\rm sign}(x)}  \frac{M_{x, \boldsymbol{\omega}}-1}{M_{x, \boldsymbol{\omega}}^2}   \right]  = 0 , \notag 
\end{align}
\noindent   $  \psi ( \cdot ) $ and $ B(\cdot)$ are defined as the above,  and   $ \psi' ( \cdot ) $ is the derivative function of $ \psi ( \cdot ) $.

Direct substitution of the above quantities  into equations (\ref{Jyuanshi})      
yields the  nonzero elements of  $ J(\boldsymbol{\omega})$ in (\ref{Jii}).

\subsection{Proof of Theorem \ref{L1234}}\label{app.l1234}    

	We only proof Equations (\ref{l2})-(\ref{l4}).  
	The 2-4th   L-moments  are  derived by using  the following  formulas     
	\begin{align}  
		\lambda_2 =& \  \int_{-\infty}^{\infty} F_{X_0}(x)\  
		( 1-F_{X_0}(x) )\  \mathrm{d} x , \label{l2hua} \\
		\lambda_3 =& \   \int_{-\infty}^{\infty} F_{X_0}(x)\  
		( 1-F_{X_0}(x)  )\   ( 2 F_{X_0}(x) -1 ) \  \mathrm{d} x ,   \label{l3hua}  \\  
		\lambda_4 =& \  \lambda_2 - 5   \Delta  
		=  \lambda_2 - 5   \int_{-\infty}^{\infty} F_{X_0} ^2(x) \ 
		( 1-F_{X_0} (x) )^2  \  \mathrm{d} x  \label{l4hua}  . 
	\end{align}

	On dividing the integration interval into $x>0$ and $x\le 0$  
	and  integrating by parts in Equations (\ref{l2hua})-(\ref{l4hua}),   
	it is easy to obtain following    expressions    		
	\begin{align}
		\lambda_2  =& \  2 \int_{-\infty}^{\infty} x f_{X_0} (x)  F_{X_0} (x) \mathrm{d} x - \int_{-\infty}^{\infty} x f_{X_0} (x)  \mathrm{d} x ,
		\label{hjl2}  \\  
		\lambda_3  =& \  6 \int_{-\infty}^{\infty} x f_{X_0} (x) 
		( F_{X_0} ^2 (x) - F_{X_0} (x) ) \mathrm{d} x  +  \int_{-\infty}^{\infty} x f_{X_0} (x)  \mathrm{d} x   \label{hjl3} ,
	\end{align}       
	\begin{equation} 
		\begin{aligned}
			\Delta   =& \  \frac{(1+r)^2}{4} \int_{0}^{\infty}      \left[   I_{u(x)} ^2 - (1+r)  I_{u(x)} ^3 + \frac{(1+r)^2}{4}   I_{u(x)} ^4 \right]  \ \mathrm{d} x    \\ 
			+& \   \frac{(1-r)^2}{4}  \int_{-\infty}^{0}  
			\left[   I_{u(x)} ^2 - (1-r)  I_{u(x)} ^3 + \frac{(1-r)^2}{4}   I_{u(x)} ^4  \right] \ \mathrm{d} x   \label{huajiandelta}  ,   
		\end{aligned} 
	\end{equation}
	
	\noindent	where $ I_{u(x)} =\  I_{u(x)} (\alpha, 1/ \beta) $ and $u(x) = \ \left[ 1+ |x|^{\beta} / \left( 2 \alpha (1+ r {\rm sign}) ^ \beta \right)\right] ^{-1}.  $

	To simplify the above equations,  we need to introduce the following  series expansions            
	\begin{align}            
		I_u (a,b)  \ =& \  \frac{u^a}{B(a,b)}\  \sum_{k=0}^{\infty}  C_k u^k  ,	 ~~~~~  
		I_u ^2 (a,b) \ =  \  \frac{ u^{2a} }{B^2(a,b)}\ \sum_{k=0}^{\infty}  D_k u^k  , \notag \\ 
		I_u ^3 (a,b) \ =& \  \frac{ u^{3a} }{B^3(a,b)}\ \sum_{k=0}^{\infty}  E_k u^k   ,  ~~~~  
		I_u ^4 (a,b) \ =  \  \frac{ u^{4a} }{B^4(a,b)}\ \sum_{k=0}^{\infty}  F_k u^k  ,  \notag   
	\end{align}       
	where $  C_k,  D_k,  E_k $  and  $ F_k $  are  defined in Theorem \ref{L1234}.

	After substituting these  series expansions in $F_{X_0}(x) $ and each power of $F_{X_0}(x) $ in  Equations  
	(\ref{hjl2})-(\ref{huajiandelta}),   
	we need to calculate the following integral formulas                      
		\begin{align}   
		\int_{0}^{\infty}  x f_{X_0} (x) I_{u(x)}  \ \mathrm{d} x  
		\ =&   \    \frac{ (2\alpha) ^ {1 / \beta} (1+r)^2 } {2 [ B(\alpha, 1 / \beta)]^2 }
		\sum_{k=0}^{\infty}  C_k  B\left(2\alpha+k-\frac{1}{\beta}, \frac{2}{\beta} \right)  ,
 \\ 
		\int_{0}^{\infty}  x f_{X_0} (x) I_{u(x)} ^2   \ \mathrm{d} x  
		\  =&   \    \frac{ (2\alpha) ^ {1 / \beta} (1+r)^2 } {2 [ B(\alpha, 1 / \beta)]^3 }
		\sum_{k=0}^{\infty}  D_k  B\left(3\alpha+k-\frac{1}{\beta}, \frac{2}{\beta} \right),
 \\  
		\int_{0}^{\infty}   I_{u(x)} ^2   \ \mathrm{d} x  
		\ =&  \    \frac{ (2\alpha) ^ {1 / \beta} (1+r) } {\beta
			[ B(\alpha, 1 / \beta)]^2 }
		\sum_{k=0}^{\infty}  D_k  B\left(2\alpha+k-\frac{1}{\beta}, \frac{1}{\beta} \right) ,
	 \\   
		\int_{0}^{\infty}   I_{u(x)} ^3   \ \mathrm{d} x  
		\  =&  \    \frac{ (2\alpha) ^ {1 / \beta} (1+r) } {\beta
			[ B(\alpha, 1 / \beta)]^3 }
		\sum_{k=0}^{\infty}  E_k  B\left(3\alpha+k-\frac{1}{\beta}, \frac{1}{\beta} \right) ,
 \\   
		\int_{0}^{\infty}   I_{u(x)} ^4   \ \mathrm{d} x  
		\ =&  \    \frac{ (2\alpha) ^ {1 / \beta} (1+r) } {\beta  [ B(\alpha, 1 / \beta)]^4 }
		\sum_{k=0}^{\infty}  F_k  B\left(4\alpha+k-\frac{1}{\beta}, \frac{1}{\beta} \right)   \label{proof6} .    
	\end{align}  
	
	\noindent  The corresponding results when   $x \le 0$ will be acquired in a similar way.   
	Plugging  the integrals above into Equations (\ref{hjl2})-(\ref{huajiandelta}),  the proof of Theorem \ref{L1234} is easily completed.       		

\subsection{Proof of Proposition \ref{ppp4}}   \label{appen.prop2}  
   
\begin{proof}
		For simplicity, we only discuss results (1) and (6). Then the results (2)-(5) and (7) follow immediately from result (1) and formula (\ref{fsgt}).   \\             
	(1). For fixed $ x, \mu, \sigma, r $ and $\beta$, note that        
	$$	\lim_{\alpha  \rightarrow \infty}   
	f_{SkeGTD} (x)   \propto   \      
	\lim_{\alpha  \rightarrow \infty} 
	\left\{  1+\frac{|x-\mu|^{\beta}}{2\alpha \sigma ^\beta [ 1+r {\rm sign} (x- \mu )]^ { \beta}  } \right\}^{-(\alpha+1/\beta)} =   \  
	\lim_{\alpha  \rightarrow \infty}
	\left(1+ \frac{ c }{\alpha}  \right)^{- \frac{\alpha}{c}\cdot c} ,  $$     
	\noindent where constant $ c=  |x-\mu|^{\beta} / \{ 2 \sigma ^\beta [ 1+r {\rm sign} (x- \mu ) ] ^ { \beta}   \} $ ,
	then we can	derive the density in (\ref{ fsgn }).    \\    
	(6). Without loss of generality, we assume $\mu =0 $ and $\sigma =1 $. Based on the result (1) of Proposition \ref{ppp4}, we only need to verify: 
	\begin{equation} 
		\lim_{\beta  \rightarrow \infty } \frac{ \beta }{ 2^{ 1+1/ \beta } \Gamma(1/\beta) } 
		\exp \left\{ - \frac{|x|^\beta}{ 2 } \right\}  = \frac{1}{2}~I\left(x<1 \right),  \label{uniform} 
	\end{equation}  
	where $I(\cdot)$ denotes the indicator function. 
	Note that 
	\begin{equation*}
		\lim_{\beta  \rightarrow \infty } \frac{ \beta }{ 2^{ 1+1/ \beta } \Gamma(1/\beta) } 
		=  \lim_{ t = \frac{1}{\beta} \rightarrow 0^+ } \frac{ 1 }{ 2^{ 1+t} \cdot t \Gamma(t) } 
		= \lim_{ t \rightarrow 0^+ } \frac{ sin(\pi t) \Gamma(1-t) }{2 \cdot \pi t} = \frac{1}{2} , 
	\end{equation*} 
	then Equation (\ref{uniform}) follows. This completes the proof of result (6) of Proposition \ref{ppp4}.       
\end{proof}

\begin{table}[H]  
	\small 
	\setlength{\abovecaptionskip}{0pt} 
	\setlength{\belowcaptionskip}{2pt}  
	\caption{ Rbias and MSE of MLEs and LMEs:  $ r= -0.5, \alpha= 4, $ and $ \beta= 2.5  $. }   \label{3can.2}          
	\setlength{\tabcolsep}{3.3mm} { 
		\begin{tabular} {lcccccc} 
			\toprule     
			\textbf{MLE}  & \multicolumn{2}{c}{ $ \bm{r =-0.5} $} & \multicolumn{2}{c}{ $\bm{\alpha = 4 }$ }  & \multicolumn{2}{c}{ $\bm{\beta = 2.5  }$ }  \\ \cmidrule{2-7}  
			\textbf{n} & $\bm{\hat{r}}$ (\textbf{Rbias})  & \textbf{MSE} & $\bm{\hat{\alpha}}$  (\textbf{Rbias})  & \textbf{MSE} & \bm{$\bm{\hat{\beta}}$} (\textbf{Rbias}) & \textbf{MSE}  \\  
			\midrule      
			20 & -0.5166 (0.2081)  &  0.0179  & 4.2942 (0.1233) & 0.2460 & 2.5159 (0.0682) & 0.0373   \\    
			50 &  -0.5026 (0.1158)  &  0.0054  & 4.2214 (0.1210) & 0.2398 & 2.5179 (0.0614) & 0.0310   \\    
			100 &  -0.5050 (0.0871)  &  0.0029  & 4.1819 (0.1201) & 0.2371 & 2.5082 (0.0611)  & 0.0207    \\  
			200 & -0.5011 (0.0535) & 0.0012 & 4.1168 (0.1192) & 0.2342 & 2.5170 (0.0478) & 0.0215 \\   
			500 & -0.5004 (0.0386) & 0.0006 & 4.0754 (0.1164) & 0.2267 & 2.5129 (0.0368) &  0.0133  \\      
			\rule{0pt}{10pt}  
			\textbf{LME} & \multicolumn{2}{c}{ $ \bm{r =-0.5} $} & \multicolumn{2}{c}{ $\bm{\alpha = 4 }$ }  & \multicolumn{2}{c}{ $\bm{\beta = 2.5 }$ }  \\ \cmidrule{2-7}  
			\textbf{n} & $\bm{\hat{r}}$ (\textbf{Rbias})  & \textbf{MSE} & $\bm{\hat{\alpha}}$  (\textbf{Rbias})  & \textbf{MSE} & \bm{$\bm{\hat{\beta}}$} (\textbf{Rbias}) & \textbf{MSE}  \\ \midrule
			20 & -0.5430 (0.2775)  & 0.0304 & 3.8325 (0.0493) & 0.0780 & 2.4267 (0.0634) & 0.0414  \\    
			50 & -0.5328 (0.1700) & 0.0114 & 3.8291 (0.0445) & 0.0483 & 2.4652 (0.0402) & 0.0388 \\   
			100 & -0.5338 (0.1265) & 0.0064 & 3.8314 (0.0438) & 0.0373 & 2.4919 (0.0294) & 0.0300 \\       
			200 & -0.5388 (0.1067) & 0.0043  & 3.8259 (0.0360) & 0.0364 & 2.4979 (0.0231) &  0.0264 \\   
			500 & -0.5330 (0.0775) & 0.0022 & 3.8236 (0.0352) & 0.0358 & 2.5068 (0.0161) & 0.0237 \\       
			\bottomrule     
	\end{tabular}      } 
\end{table}

\begin{table}[H]  
	\small   
	\setlength{\abovecaptionskip}{0pt}  
	\setlength{\belowcaptionskip}{2pt} 
	\caption{ Rbias and MSE of MLEs and LMEs:  $ r= 0.7, \alpha= 3, $ and $ \beta= 2.5  $.   }           
	\setlength{\tabcolsep}{3.3mm} {  
		\label{3can.5}  
		\begin{tabular} {lcccccc} 
			\toprule      
			\textbf{MLE}  & \multicolumn{2}{c}{ $ \bm{r = 0.7} $} & \multicolumn{2}{c}{ $\bm{\alpha = 3 }$ }  & \multicolumn{2}{c}{ $\bm{\beta = 2.5  }$ }  \\ \cmidrule{2-7}  
			\textbf{n} & $\bm{\hat{r}}$ (\textbf{Rbias})  & \textbf{MSE} & $\bm{\hat{\alpha}}$  (\textbf{Rbias})  & \textbf{MSE} & \bm{$\bm{\hat{\beta}}$} (\textbf{Rbias}) & \textbf{MSE}  \\  
			\midrule     
			20 & 0.6501 (0.1105)  &  0.0129  & 3.2641 (0.2884) & 0.8026 & 2.5384 (0.0573) & 0.0310    \\    
			50 &  0.6895 (0.0699)  &  0.0035  & 3.2867 (0.2705) & 0.7296 & 2.5109 (0.0541) & 0.0295   \\           
			100 & 0.7041 (0.0498)  &  0.0018  & 3.2423 (0.2506) & 0.6519 & 2.5127 (0.0538)  & 0.0221    \\       
			200 & 0.7016 (0.0361) & 0.0010 & 3.2355 (0.2338) & 0.5878 & 2.5184 (0.0530) & 0.0218    \\  
			500 & 0.7007 (0.0213) & 0.0003 & 3.1302 (0.2168) & 0.5203 & 2.5215 (0.0502) &  0.0203  \\   
			\rule{0pt}{10pt}   
			\textbf{LME} & \multicolumn{2}{c}{ $ \bm{r = 0.7} $} & \multicolumn{2}{c}{ $\bm{\alpha = 3 }$ }  & \multicolumn{2}{c}{ $\bm{\beta = 2.5 }$ }  \\ \cmidrule{2-7}  
			\textbf{n} & $\bm{\hat{r}}$ (\textbf{Rbias})  & \textbf{MSE} & $\bm{\hat{\alpha}}$  (\textbf{Rbias})  & \textbf{MSE} & \bm{$\bm{\hat{\beta}}$} (\textbf{Rbias}) & \textbf{MSE}  \\
			\midrule 
			20 & 0.7417 (0.1561)  &  0.0167  & 2.8115 (0.0988) & 0.1472 & 2.3034 (0.0935) & 0.0764   \\    
			50 &  0.7548 (0.1222)  &  0.0114  & 2.7797 (0.0915) & 0.0989 & 2.3259 (0.0736) & 0.0431   \\    
			100 & 0.7708 (0.1173)  &  0.0104  & 2.7331 (0.0910) & 0.0855 & 2.3306 (0.0691)  & 0.0342    \\       
			200 & 0.7702 (0.1079) & 0.0088 & 2.7254 (0.0821) & 0.0780 & 2.3394 (0.0660) & 0.0297    \\  
			500 & 0.7791 (0.1034) & 0.0083 & 2.6922 (0.0726) & 0.0746 & 2.3459 (0.0637) &  0.0269  \\            
			\bottomrule     
	\end{tabular}   }      
\end{table}

\begin{table}[H]   
	\setlength{\abovecaptionskip}{0pt}  
	\setlength{\belowcaptionskip}{2pt}  
	\caption{ Rbias and MSE of TSE of  $(\mu, \sigma, r, \alpha, \beta)= (4, 2, 0.7, 4, 2) $ for the SkeGTD.   }           
		\setlength{\tabcolsep}{0.83mm} { 
	    \scriptsize  
		\label{5can.3} 
		\begin{tabular} {lcccccccccc} 
			\toprule      
			\textbf{TSE} & \multicolumn{2}{c}{ $ \bm{\mu = 4} $}   & \multicolumn{2}{c}{ $ \bm{\sigma = 2} $}  &   \multicolumn{2}{c}{ $ \bm{r = 0.7} $} & \multicolumn{2}{c}{ $\bm{\alpha = 4 }$ }  & \multicolumn{2}{c}{ $\bm{\beta = 2 }$ }  \\   \cmidrule{2-11}      
			\textbf{n}  & $\bm{\hat{\mu}}$ (\textbf{Rbias})  & \textbf{MSE}    & $\bm{\hat{\sigma}}$ (\textbf{Rbias})  & \textbf{MSE}    & $\bm{\hat{r}}$ (\textbf{Rbias})  & \textbf{MSE} & $\bm{\hat{\alpha}}$  (\textbf{Rbias})  & \textbf{MSE} & \bm{$\bm{\hat{\beta}}$} (\textbf{Rbias}) & \textbf{MSE}  \\  
			\midrule        
			20 & 4.0110	(0.0505) &	0.0533 &	2.2731	(0.2143) &	0.2751 &	0.7102	(0.1780) &	0.0258 &	4.4184	(0.1057) &	0.1976 &	2.3709	(0.2065) &	0.1925  \\     
			50   & 3.6950	(0.0856) &	0.1654 &	2.2472	(0.1757) &	0.1868 &	0.8194	(0.1947) &	0.0250  &	4.3891	(0.1008) &	0.1832 &	2.3294	(0.1940) &	0.1773 \\ 
			100  & 3.7540	(0.0661)  &	0.0803 & 	2.2216	(0.1683)  & 	0.1614 &	0.7990	(0.1503) &	0.0144 &	4.3500	(0.0936) &	0.1643 &	2.2710	(0.1816) &	0.1597 \\ 
			200  &  3.5909	(0.1034) &	0.1856 &	2.1460	(0.1508) &	0.1257 &	0.8470	(0.2108) &	0.0239 &	4.3422	(0.0786) &	0.1407 & 	2.2040	(0.1559) &	0.1276 \\ 
			500  & 3.5505	(0.1124) &	0.2081 &	2.0939	(0.1263) &	0.0929 &	0.8606	(0.2294) &	0.0265 &	4.3261	(0.0660) &	0.1032	 &  2.1538	(0.1262) &	0.0895  
			\\        
			\bottomrule     
		\end{tabular}        } 
	\end{table}

	\begin{table}[H]      
		\setlength{\abovecaptionskip}{0pt}  
		\setlength{\belowcaptionskip}{2pt}  
		\caption{ Rbias and MSE of TSE  of  $(\mu, \sigma, r, \alpha, \beta)= (-4, 2, 0.7, 4, 2) $ for the SkeGTD.   }           
			\setlength{\tabcolsep}{0.83mm} {  
		    \scriptsize 
			\label{5can.4} 
			\begin{tabular} {lcccccccccc} 
				\toprule      
				\textbf{TSE} & \multicolumn{2}{c}{ $ \bm{\mu = -4} $}   & \multicolumn{2}{c}{ $ \bm{\sigma = 2} $}  &   \multicolumn{2}{c}{ $ \bm{r = 0.7} $} & \multicolumn{2}{c}{ $\bm{\alpha = 4 }$ }  & \multicolumn{2}{c}{ $\bm{\beta = 2 }$ }  \\   \cmidrule{2-11}      
				\textbf{n}  & $\bm{\hat{\mu}}$ (\textbf{Rbias})  & \textbf{MSE}    & $\bm{\hat{\sigma}}$ (\textbf{Rbias})  & \textbf{MSE}    & $\bm{\hat{r}}$ (\textbf{Rbias})  & \textbf{MSE} & $\bm{\hat{\alpha}}$  (\textbf{Rbias})  & \textbf{MSE} & \bm{$\bm{\hat{\beta}}$} (\textbf{Rbias}) & \textbf{MSE}  \\  
				\midrule       
				20 & 	-4.1444	(0.0670)	&  0.1050 & 	2.2608	(0.2153) & 	0.2766	&   0.7805	(0.2056) & 	0.0311 & 	4.4218	(0.1069) & 	0.2009 & 	2.3765	(0.2100) & 	0.1969 \\
				50 & 	-4.2926	(0.0848) & 	0.1329	 &  2.2341	(0.1729) & 	0.1800 & 	0.8313	(0.2006) & 	0.0259 & 	4.3941	(0.1014) & 	0.1855 & 	2.3356	(0.1963) & 	0.1791 \\
				100 & 	-4.3859	(0.1002) & 	0.1779 & 	2.2090	(0.1620) & 	0.1531 & 	0.8487	(0.2146) & 	0.0255 & 	4.3674	(0.0961) & 	0.1704 & 	2.2934	(0.1852) & 	0.1636 \\
				200 & 	-4.4464	(0.1118) & 	0.2148 & 	2.1713	(0.1447) & 	0.1148 & 	0.8652	(0.2361) & 	0.0248 & 	4.3460	(0.0900) & 	0.1530  & 	2.2605	(0.1677) & 	0.1416 \\
				500 & 	-4.4942	(0.1236) & 	0.2536 & 	2.1173	(0.1287) & 	0.0914 & 	0.8736	(0.2480) & 	0.0108 & 	4.3154	(0.0825) & 	0.1305 & 	2.2103	(0.1451) & 	0.1119 
				\\        
				\bottomrule     
			\end{tabular}    }         
		\end{table}  

\begin{table}[htbp]     
	\centering    	\small 	    
	 \begin{threeparttable}
	\caption{ Percentages of selecting the SkeGTD in 1000 simulations according to the AIC, BIC, and EDC criteria.   }       
	\label{tablesimu2}        
	\setlength{\tabcolsep}{1.49mm}{      
		\begin{tabular}{lccccc}    
			\toprule \\   
			\textbf{criteria} & \textbf{n}  & \textbf{ SkeGTD vs N } & \textbf{ SkeGTD vs $t$ } & \textbf{ SkeGTD vs SN}  & \textbf{ SkeGTD vs ST }  
			\\  \midrule   
			AIC(\%)      & 20  & 90.8        & 80.2        & 73.4         & 55.3         \\
			& 50  & 99.5        & 78.8        & 92.5         & 68.5         \\
			& 100 & 99.9        & 80.5        & 99.2         & 74.0         \\
			& 200 & 100         & 82.1        & 100          & 71.0         \\
			& 500 & 100         & 89.9        & 100          & 76.4         \\
			\addlinespace[4pt]  
			BIC(\%)     & 20  & 86.9        & 79.8        & 66.4          & 45.3         \\
			& 50  & 99.3        & 78.1        & 89.1         & 62.2         \\
			& 100 & 99.9        & 80.1        & 98.5         & 69.6         \\
			& 200 & 100         & 81.2        & 100          & 67.5         \\
			& 500 & 100         & 89.7        & 100          & 73.2         \\ 
			\addlinespace[4pt]     
			EDC(\%)     & 20  & 93.5        & 80.7        & 78.0          & 63.8         \\
			& 50  & 99.5        & 78.9        & 93.3         & 70.1         \\
			& 100 & 99.9        & 80.5        & 99.2         & 74.0         \\
			& 200 & 100         & 82.0        & 100          & 70.0         \\
			& 500 & 100         & 89.7        & 100          & 75.0        
			\\ 	\bottomrule      
		\end{tabular}   
	}    
\begin{tablenotes}
	\item[]  AIC(\%), BIC(\%), and EDC(\%) denote the SkeGTD's best performance percentages.   
\end{tablenotes}   
	\end{threeparttable} 
\end{table}

\begin{table}[H]   
	\centering      
	\setlength{\abovecaptionskip}{5pt} 
	\setlength{\belowcaptionskip}{2pt}  
	\caption{ A description  of Roller data. n, $\sqrt{b_1}$ and $b_2$ represent sample  size,  sample skewness and kurtosis coefficients,  respectively. }        
	\label{summary roller}   
	\setlength{\tabcolsep}{6mm} {   
		\begin{tabular} {cccccccc} 
			\toprule      
			\bm{$n$} & \bm{$\bar{Y}$} & \bm{$S^2$ } &   \bm{$\sqrt{b_1}$} &  \bm{ $b_2$}  & \textbf{min(\bm{$Y$}) } & \textbf{max(\bm{$Y$}) }  \\     
			1150 &  3.535 & 0.422 & -0.988 & 1.877 &  0.237  & 5.150   \\
			\bottomrule          
	\end{tabular}     }         
\end{table}

\begin{table}[H]   
	\centering    	\small 	  
	\setlength{\abovecaptionskip}{2pt} 
	\setlength{\belowcaptionskip}{2pt}   
	\caption{ Estimated parameters, log-likelihood and the AIC values for SN,  SFN,  ASGT and SkeGTD models in roller datasets. }          
	\label{aic}    
	\setlength{\tabcolsep}{0.01mm} {  
		\begin{tabular} {lccccccc}  
			\toprule   
			\textbf{Models } & \bm{$\hat{\mu}$  } & \bm{$\hat{\sigma}$} & \bm{$\hat{r}$} & \bm{ $\hat{\alpha}$} & \bm{ $\hat{\beta}$ } & \textbf{LogL} & \textbf{ AIC} \\  
			\midrule   
			SN & 4.248 & 0.964 & - & -2.758 & - & -1071.347 \  & 2148.694  
			\\  
			SFN & 3.976   &  1.701  & - & -2.743  &  2.310 & -1064.348 \  &  2136.696   
			\\
			ASGT & 3.784  &   0.621  &  0.3877 & 1.5208   &  5.8424 & -1072.182 \  & 2154.365   
			\\  		
			\specialrule{0em}{2pt}{2pt} 
			SkeGTD & 3.837 & 0.446  & -0.322  &  10.000     
			&  1.565&  -1062.446  \ 
			&  2134.891
			\\    
			(BCI)   &   
			(3.776,3.899)      &  
			(0.372,0.519)      &
			(-0.391,-0.252)   &
			(5.819,14.181)    &   
			(1.280,1.850)     & -  &  -   
			\\   			
			\bottomrule       
	\end{tabular}      } 
\end{table}

\begin{figure}[H] 	   
	\vspace{0cm}  
	\setlength{\abovecaptionskip}{-20pt} 
	\subfigtopskip=2pt  
	\subfigbottomskip=30pt  
	\subfigcapskip=0pt 
	\subfigure[] { 
		\includegraphics[width=5 cm]{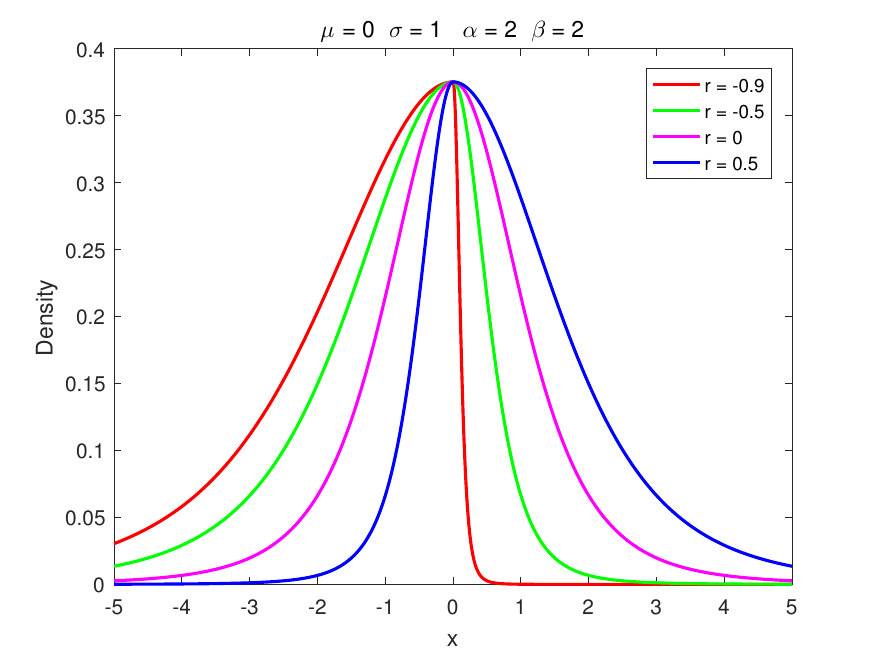}  \label{fig1} 
	}   
	\subfigure[] {  
		\includegraphics[width=5 cm]{1118rr-eps-converted-to.pdf}  \label{fig2} 
	}  
	\subfigure [] {    
		\includegraphics[width=5 cm ]{1118rr-eps-converted-to.pdf}  \label{fig3}   
	}     
	\caption{ Probability density functions of the SkeGTD for selected values of the parameters $r, \alpha $ and $\beta$. 
		(a)  Effects of skewness parameter $r$ with fixed  $ \mu=0, ~ \sigma=1,~ \alpha=2,~ \beta=2$; (b)  Effects of shape parameter $\alpha$ with fixed $\mu=0,~ \sigma=1,~ r=0.5,~ \beta=2 $; (c) Effects of shape parameter $\beta$ with fixed $\mu=0,~ \sigma=1,~ r=0.5,~ \alpha=2 $.      
	}    
	\label{plotf} 
\end{figure}

\begin{figure}[htb]     	
	\centering 
	\vspace{0cm}  
	\setlength{\abovecaptionskip}{-20pt}  
	\subfigtopskip=2pt  
	\subfigbottomskip=30pt  
	\subfigcapskip=0pt  
	\subfigure[]{  
		\includegraphics[width=3.7cm,height=2.8cm ]{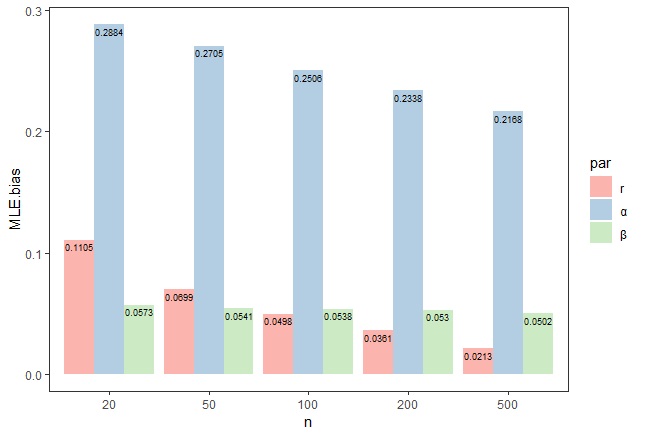}  \label{bar1} }    
	\subfigure[]{    
		\includegraphics[width=3.7cm,height=2.8cm ]{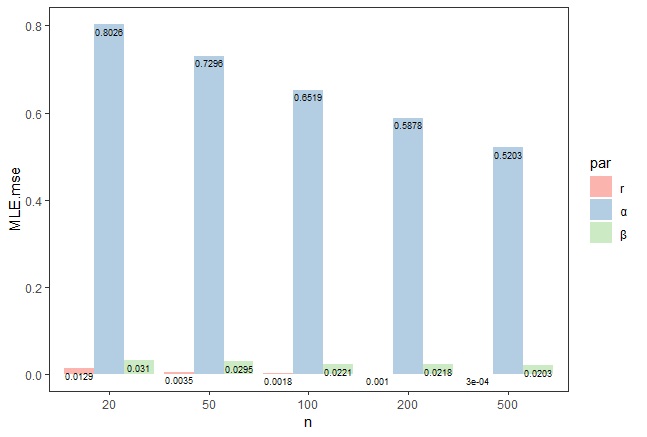} 
		\label{bar2}    
	} 
	\subfigure[]{   
		\includegraphics[width=3.7cm,height=2.8cm ]{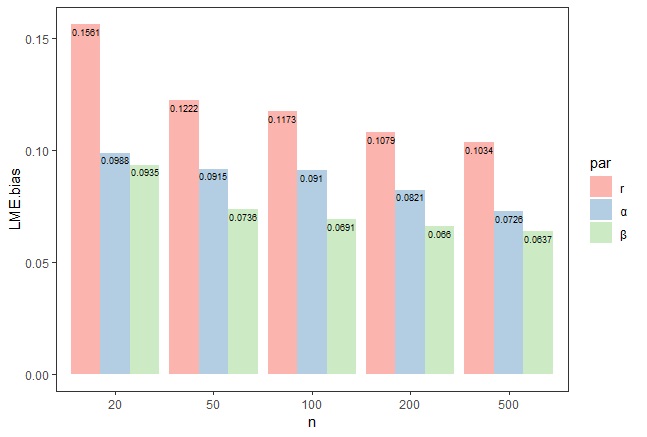}  \label{bar3}  
	}   
	\subfigure[]{   
		\includegraphics[width=3.7cm,height=2.8cm ]{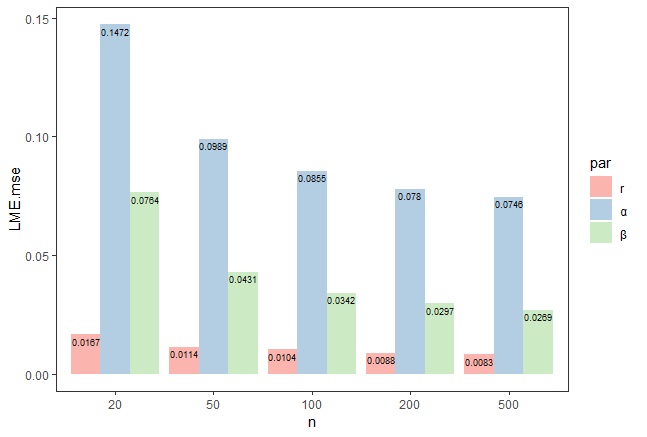}  \label{bar4}  
	} 
	\caption{ Rbias and MSE of MLEs and LMEs of $ (r, \alpha, \beta)  (r = 0.7, \alpha = 3, \beta = 2.5)  $.  
		(a)  Rbias of the MLE; 
		(b)  MSE of the MLE; (c) Rbias of the LME; (d) MSE of the LME.     	   
	}  	\label{bar}     
\end{figure}

\begin{figure}[H]  
	\centering    
	\includegraphics[width=8 cm]{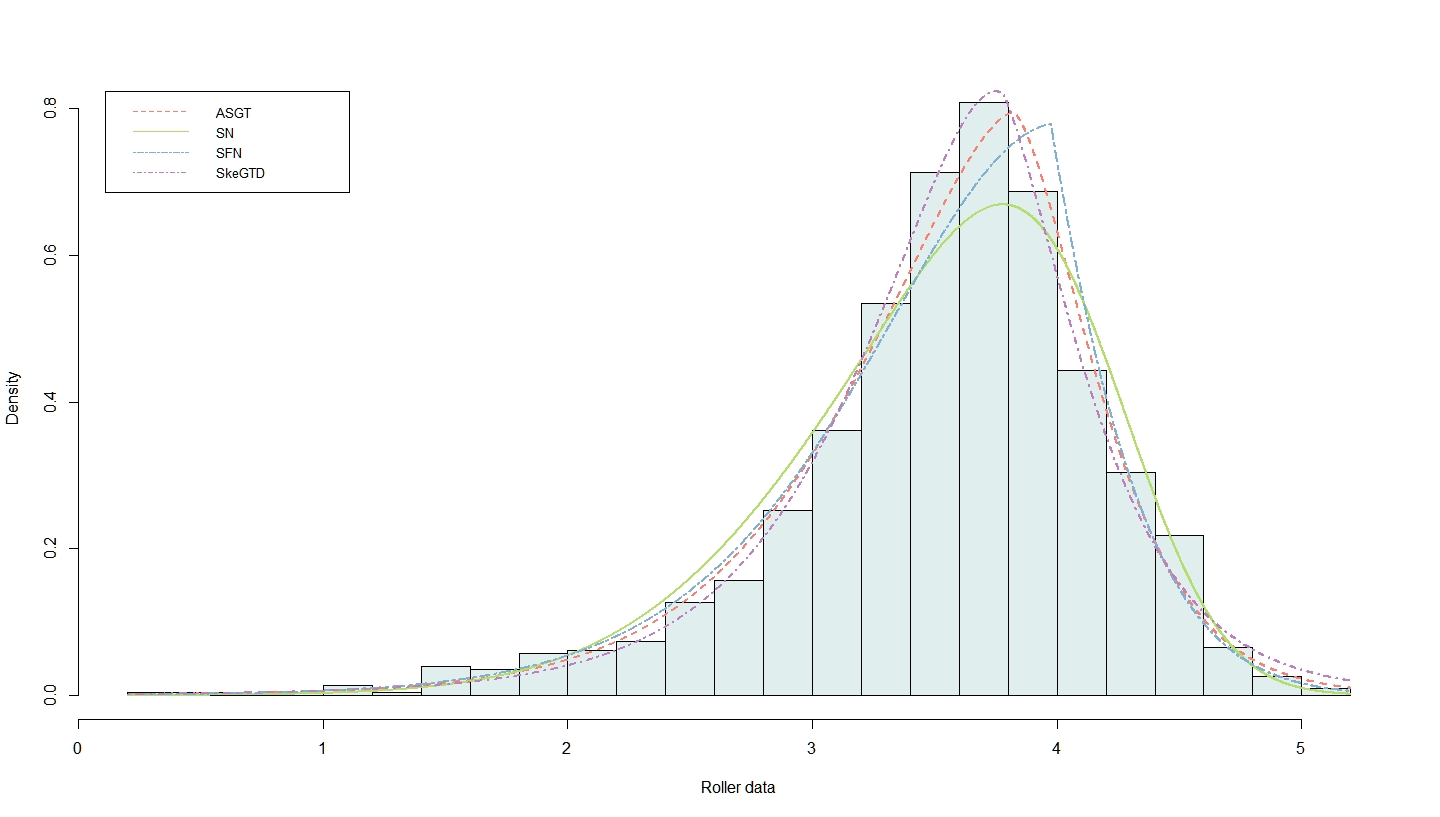}    
	\caption {  Histogram of the roller data,  together with the fitted SN, SFN, ASGT and SkeGTD  densities.
	}         
	\label{roller}    
\end{figure}

\begin{figure}[H] 	  
	\centering   
	\vspace{0cm}    
	\setlength{\abovecaptionskip}{-20pt}  
	\subfigtopskip=2pt   
	\subfigbottomskip=30pt  
	\subfigcapskip=0pt 
	\subfigure[] {  
		\includegraphics[width=6cm]{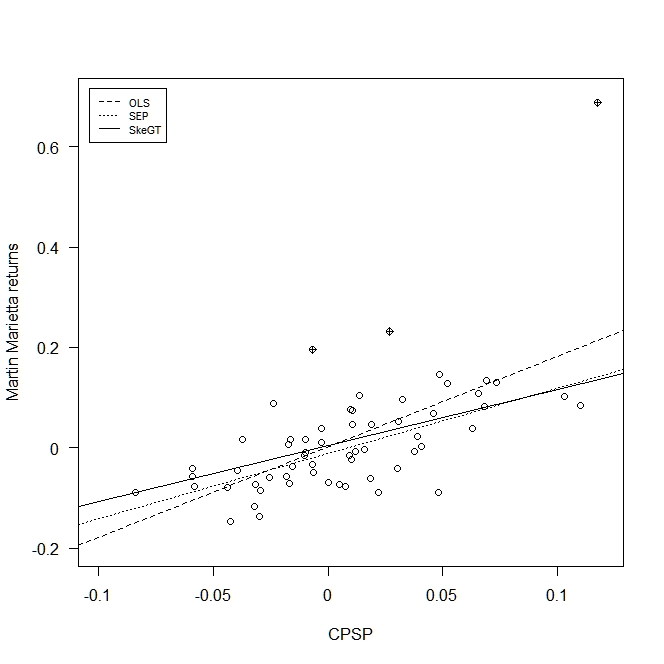} \label{fig:scar1} 
	}  
	\subfigure[] {  
		\includegraphics[width=6cm]{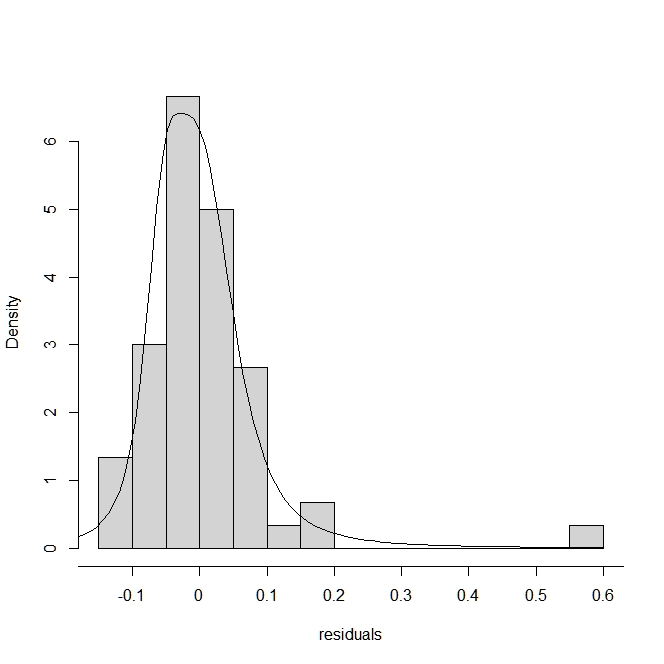}
		\label{fig:hist1}  
	}     
	\caption{  Martin Marietta data. (a) Scatter plot and fitted regression lines.  (b)   Histogram of the residuals with the fitted SkeGTD density.       }   \label{MMdata}                
\end{figure}

\subsection*{\centering List of Figures}  

	\begin{compactitem}  
		\item   	Figure 1:    
		Probability density functions of the SkeGTD for some selected values of the parameters $r,~ \alpha $ and $\beta$. 
		(a)  Effects of skewness parameter $r$ with fixed  $ \mu=0, ~ \sigma=1,~ \alpha=2,~ \beta=2$; (b)  Effects of shape parameter $\alpha$ with fixed $\mu=0,~ \sigma=1,~ r=0.5,~ \beta=2 $; (c) Effects of shape parameter $\beta$ with fixed $\mu=0,~ \sigma=1,~ r=0.5,~ \alpha=2 $. 
		\item 			Figure 2: 
		Rbias and MSE of MLEs and LMEs of $ (r, \alpha, \beta)  (r = 0.7, \alpha = 3, \beta = 2.5)  $.  
		(a)  Rbias of the MLE; 
		(b)  MSE of the MLE; (c) Rbias of the LME; (d) MSE of the LME.  
		\item 			Figure 3: 
		Histogram of the roller data,  together with the fitted SN, SFN, ASGT and SkeGTD  densities.  
		\item 			Figure 4: 
		Martin Marietta data. (a) Scatter plot and fitted regression lines.  (b)   Histogram of the residuals with the fitted SkeGTD density.  
	\end{compactitem}

\end{document}